\documentclass[journal,preprint]{IEEEtran}

\usepackage{cite}

\usepackage{ifpdf}
\ifpdf
    \usepackage{graphicx}
    \usepackage{epstopdf}
\else
    \usepackage{graphicx}

\fi

\usepackage{mathptmx} % assumes new font selection scheme installed
\usepackage{times} % assumes new font selection scheme installed
\usepackage{amsmath} % assumes amsmath package installed
\usepackage{amssymb}  % assumes amsmath package installed
\usepackage{tabularx}
\usepackage{algorithm}
\usepackage{algorithmic}
\usepackage{ulem}
\usepackage{dsfont}
\usepackage[dvipsnames,usenames]{color}
\usepackage[colorlinks,%
linkcolor=BrickRed,%
filecolor=BrickRed,%
citecolor=RoyalPurple,%
]{hyperref}
\usepackage{caption}
\usepackage{subfigure}
\usepackage{color}

\usepackage{tabularx}
\renewcommand{\k}{{\sf K}}
\newcommand{\K}{{\sf K}}

\newcommand{\newP}[1]{\noindent{\bf #1:}}

\newcommand{\ud}{\mathrm{d}}
\def\Re{\mathbb{R}}

\def\Proposition#1{Prop.~\ref{#1}}

\def\Sec#1{Sec.~\ref{#1}}

\newcommand{\tr}{\mbox{tr}}
\def\eqdef{\mathrel{:=}}

\def\clZ{{\cal Z}}

\def\E{{\sf E}}
\def\expect{{\sf E}}
\def\Expect{{\sf E}}

\def\P{{\sf P}}

\newtheorem{theorem}{Theorem}
\newtheorem{example}{Example}
\newtheorem{definition}{Definition}
\newtheorem{lemma}{Lemma}
\newtheorem{remark}{Remark}
\newtheorem{proposition}{Proposition}
\newtheorem{corollary}{Corollary}
\newcommand{\trace}{\text{Tr}}

\newcommand{\mN}{m^{(N)}}
\newcommand{\SigN}{\Sigma^{(N)}}

\newcommand{\Fnorm}[1]{\|#1\|_F}

\newcounter{rmnum}

\newenvironment{romannum}{\begin{list}{{\upshape (\roman{rmnum})}}{\usecounter{rmnum}
			\setlength{\leftmargin}{12pt}
			\setlength{\rightmargin}{8pt}
			\setlength{\itemsep}{2pt}
			\setlength{\itemindent}{-1pt}
	}}{\end{list}}

\newcounter{anum}

\newcommand{\Xbar}{\bar{X}}
\newcommand{\PP}{{\sf P}}
\newcommand{\NN}{\mathcal{N}}
\def\Kbar{\bar{\sf K}}
\newcommand{\mbar}{\bar{m}}

\newcommand{\Sigmabar}{\bar{\Sigma}}

\newcommand{\X}{X}
\newcommand{\sRicc}{\sqrt{\text{Ricc}}}
\newcommand{\Ricc}{\text{Ricc}}
\newcommand{\ii}{k}
\def\FRAC#1#2#3{\genfrac{}{}{}{#1}{#2}{#3}}
\def\half{{\mathchoice{\FRAC{1}{1}{2}}%
		{\FRAC{2}{1}{2}}%
		{\FRAC{3}{1}{2}}%
		{\FRAC{4}{1}{2}}}}

\newcommand{\calN}{\mathcal{N}}

\newcommand{\KN}{{\sf K}^{(N)}}
\newcommand{\Bbar}{\bar{B}}

\newcommand{\xibar}{\bar{\xi}}

\newcommand{\ddt}{\frac{\ud}{\ud t}}

\title{\LARGE \bf
An Optimal Transport Formulation of the Ensemble Kalman Filter}

\author{Amirhossein Taghvaei,  Prashant G. Mehta% <-this % stops a space
\thanks{A. Taghvaei is with the Department of Mechanical and Aerospace
  Engineering at University of California Irvine.  The research
  reported in this paper was performed while he was a graduate student
  at the  University of
	Illinois at Urbana-Champaign {\tt\small ataghvae@uci.edu}.}
\thanks{P. G. Mehta is with the Coordinated Science Laboratory and
the Department of Mechanical Science and Engineering at the University of
Illinois at Urbana-Champaign (UIUC)
 {\tt\small mehtapg@illinois.edu}.}% <-this % stops a space,
 \thanks{Financial support from the NSF CMMI grants  1462773 and
   1761622 is gratefully acknowledged.}%
 \thanks{The conference version of this paper appears in~\cite{AmirACC2016} and~\cite{AmirACC2018}.}%
}

\begin{document}
\normalem
\maketitle

\begin{abstract}
	Controlled interacting particle systems such as the ensemble Kalman
	filter (EnKF) and the feedback particle filter (FPF) are numerical
	algorithms to approximate the solution of the nonlinear filtering
	problem in continuous time.  The distinguishing feature of these
	algorithms is that the Bayesian update step is implemented using a
	feedback control law.  
	It has been noted in the literature that the control
	law is not unique.  This is the main problem addressed in this
	paper.  To obtain a unique control law, the filtering problem is
	formulated here as an optimal transportation problem.  An explicit formula
	for the (mean-field type) optimal control law is derived in the linear
	Gaussian setting.  Comparisons are made with the control laws
        for different types of EnKF
	algorithms described in the literature.  Via empirical approximation
	of the mean-field control law, a finite-$N$ controlled interacting particle algorithm is obtained.  For this algorithm, the equations for
	empirical mean and covariance are derived and shown to be identical to
	the Kalman filter.  This allows strong conclusions on convergence and
	error properties based on the classical filter stability theory for
	the Kalman filter.  It is shown that, under certain technical conditions, the mean squared error (m.s.e.) converges to zero {\em
		even} with a finite number of particles. %  For the
        %       m.s.e., a polynomial scaling is derived with respect to
        %       the
	% state dimension. 
        A detailed propagation of chaos
	analysis is carried out for the finite-$N$ algorithm. The analysis is used to prove weak
	convergence of the empirical distribution as $N\rightarrow\infty$.  
	For a certain
	simplified filtering problem, analytical comparison of the m.s.e. with
	the importance sampling-based algorithms is described.  The
        analysis helps explain the favorable scaling properties of
        the control-based algorithms reported in
        several numerical studies in recent literature. 
\end{abstract}

\section{Introduction}
\label{sec:intro}

The subject of this paper concerns Monte-Carlo methods for simulating a
nonlinear filter (conditional distribution) in continuous-time
settings.  
The mathematical abstraction of any filtering problem involves two
processes: a hidden Markov process $\{X_t\}_{t\ge
	0}$ and the observation process $\{Z_t\}_{t\ge
	0}$.  The numerical 
problem is to compute the {\em posterior distribution} $\PP(X_t \in \cdot  | \clZ_t)$
where $\clZ_t:=\sigma\{Z_s;0 \leq s \leq t\}$ is the filtration
generated by the observations.  A standard solution approach % for
% the Monte-Carlo approximation of the posterior distribution
is the
particle filter which relies on importance sampling to implement the
effect of conditioning~\cite{gordon93,doucet09}.  In numerical
implementations, this often leads to the particle degeneracy issue (or
weight collapse) whereby only a few particles have large weights.
To combat this issue, various types of resampling schemes have been
proposed in the literature~\cite{bain2009,delmoralbook}.

In the past decade, an alternate class of algorithms has attracted
growing attention.  These algorithms can be regarded as a controlled
interacting particle system where the central idea is to implement the
effect of conditioning using feedback control.  Mathematically, this
involves construction of controlled stochastic process, denoted by
$\{\Xbar_t\}_{t\ge 0}$.  In continuous-time settings, the model for
the $\bar{X}_t$ is a stochastic differential equation (sde):
\begin{equation}
\ud \Xbar_t = u_t(\Xbar_t)\ud t + \k_t(\Xbar_t) \ud Z_t +
\text{[additional terms]},\quad \Xbar_0 \stackrel{\text{d}}{=} X_0
\label{eq:u-k}
\end{equation}
where the [additional terms] are pre-specified (these terms may be
zero).  % If present, they may
% include terms to simulate dynamics of $X$ or noise.
The control problem is to
design mean-field type control law $\{u_t(\cdot)\}_{t\ge
	0}$ and $\{\k_t(\cdot)\}_{t\ge
	0}$ such that
the conditional distribution of $\bar{X}_t$ (given $\clZ_t$) is equal
to the posterior distribution of $X_t$.  If this property holds, the
filter is said to be {\em exact}.  In a numerical implementation, the
mean-field terms in the control law are approximated empirically by simulating $N$
copies of~\eqref{eq:u-k}.  The resulting system is a controlled
interacting particle system with a finite number of $N$ interacting
particles.  
The particles have uniform importance weights by
construction.  Therefore, the particle degeneracy issue does not
arise.  Resampling is no longer necessary and steps such as rules for reproduction, death or birth of
particles are altogether avoided. 

The focus of this paper is on (i) formal methods for design
of control laws ($u_t(\cdot)$
and $\k_t(\cdot)$) for~\eqref{eq:u-k}; (ii) algorithms for empirical approximation of the
control laws using $N$ particles; and (iii) error analysis of the
finite-$N$ interacting particle models as $N\rightarrow\infty$.   
The main problem highlighted and addressed in this paper is the issue
of uniqueness:  
one can interpret the controlled system~\eqref{eq:u-k} as transporting
the initial distribution at time $t = 0$ (prior) to the conditional
distribution at time t (posterior).  
Clearly, there are infinitely many maps that transport one
distribution into another. 
This suggests that there are infinitely many choices of control laws that all lead to exact filters.  
This is not surprising: The exactness condition specifies only the
marginal distribution of the stochastic process $\{\Xbar_t\}_{t\geq 0}$
at times $t\geq 0$, which is not enough to uniquely identify a
stochastic process, e.g., the joint distributions at two time instants
are not specified. 

Although these issues are relevant more generally, the scope of this
paper is limited to the linear Gaussian problem.  A motivation comes
from the widespread use of the ensemble Kalman filter (EnKF) algorithm
in applications.  It is noted that the mean-field limit of the EnKF
algorithm is exact only in linear Gaussian settings.  The issue of non-uniqueness is manifested in the
different types of EnKF algorithms reported in literature.  Some of
these EnKF types are discussed as part of the literature survey
(in~\Sec{sec:lit-survey}) and in the main body
of the paper (in~\Sec{sec:non-uniqueness-subsec}).

\subsection{Contributions of this paper}  

The following  contributions are made in this paper:
\begin{enumerate}
	\item {\bf Non-uniqueness issue:} For the linear Gaussian problem, an error process is introduced to
	help explain the non-uniqueness issue in the selection of the
	control law in~\eqref{eq:u-k}.  The error process helps clarify the 
	relationship between the different types of control laws leading to
	the different types of EnKF algorithms that have appeared over the
	years in the literature.  
	% relate classical
	%   EnKF with perturbed observation and square-root EnKF algorithm as
	%   well as suggest other types of deterministic filters.   
	\item {\bf Optimal transport FPF:}  To select a unique
	control law, an optimization problem is proposed in the form of a
	time-stepping procedure.  The optimality concept is motivated by the optimal
	transportation theory~\cite{evans,villani2003}. 
	The solution of the time-stepping procedure yields a unique optimal control law.
	%denoted as $u^*_t $ and $\k^*_t$.  
	The resulting filter is referred to as
	the optimal transport FPF.  The procedure is suitably adapted
	to handle the case, important in Monte-Carlo 
	implementations with finitely many $N$ particles, where the covariance is singular.  In this case,
	the optimal (deterministic) transport maps are replaced by optimal
	(stochastic) couplings.  The general form of the optimal FPF
	includes stochastic terms which are zero when the covariance is
	non-singular.   %  In this case, one obtains the deterministic form of

	\item  {\bf Error analysis:} A detailed error analysis is carried out
	for the deterministic form of the optimal FPF for the finite but
	large $N$ limit.  For the purposes of error analysis, it is 
	assumed that the linear system is controllable and observable, and
	the initial empirical covariance matrix is non-singular.  The main
	results are as follows:  
	\begin{romannum}
		\item Empirical mean and covariance of particles
		is shown to converge almost surely to exact mean and covariance
		as $t\rightarrow\infty$ even for finite $N$
		(Prop.~\ref{prop:conv_error}-(i)); 
		\item Mean-squared error is shown to be bounded by 
		$\frac{Ce^{-\lambda t}}{\sqrt{N}}$ where the constant $C$ 
		has polynomial dependence on the problem dimension (Prop.~\ref{prop:conv_error}-(ii)); 
		\item A propagation of chaos analysis is carried out to show that
		empirical distribution of the particles converges in a weak sense to the exact filter
		posterior distribution (Cor.~\ref{cor:prop-chaos}).
	\end{romannum}
	%\pgm{AMIR: Be precise!  This bullet needs to be completely re-written.}
	\item  {\bf Comparison to importance sampling:}  For a certain
	simplified filtering problem, a comparison of the m.s.e. between the importance
	sampling and control-based filters is described.  The main result is
	to show that using an important sampling approach, the number of
	particles $N$ must grow exponentially with the dimension $d$.  In contrast, with a control-based approach, $N$ scales at  most as
	order $d^2$ in order to maintain the same error (Prop.~\ref{prop:importance-sampling}).  The conclusions are also
	verified numerically (Fig.~\ref{fig:error-PF-FPF}).    

\end{enumerate}

%\subsection{Comparison with the conference papers~\cite{AmirACC2016,AmirACC2018}}
This paper extends and completes the preliminary results reported in our prior conference papers~\cite{AmirACC2016,AmirACC2018}. 
The optimal transport formulation of the FPF and the time-stepping
procedure was originally introduced in~\cite{AmirACC2016}.  However,
its extension to the singular covariance matrix case,
in~\Sec{sec:singular-case}, is original and has not appeared
before. Preliminary error analysis of the deterministic form of
optimal FPF appeared in our prior work~\cite{AmirACC2018}.  The
current paper extends the results in~\cite{AmirACC2018} in two key
aspects: (i) The error bounds for the convergence of the empirical
mean and covariance, in Prop.~\ref{prop:conv_error}-(ii), reveal the
scaling with the problem dimension~(see
Remark~\ref{rem:scaling-dim}-(ii)), whereas previous results did not;
(ii) The propagation of chaos analysis  is carried out for the vector
case~(Cor.~\ref{cor:prop-chaos}), whereas previous result was only
valid for the scalar case. These improvements became  possible with a
proof approach that is entirely different than the one used
in~\cite{AmirACC2018}. Finally, the analytical comparison with the importance
sampling particle filter, in Sec.~\ref{sec:PF}, is new.

\subsection{Literature survey}\label{sec:lit-survey}

Two examples of the controlled interacting particle systems are the
classical ensemble Kalman filter
(EnKF)~\cite{evensen1994sequential,evensen2003ensemble,whitaker2002ensemble,Reich-ensemble}
and the more recently developed feedback particle filter
(FPF)~\cite{taoyang_TAC12,yang2016}.  The EnKF algorithm is the
workhorse in applications (such as weather prediction) where the state
dimension $d$ is very high; cf.,~\cite{Reich-ensemble,houtekamer01}.  The high-dimension of the
state-space provides a significant computational challenge {\em even}
in linear Gaussian settings.  For such problems, an EnKF
implementation may require less computational resources (memory and
FLOPS) than a Kalman filter~\cite{houtekamer01,evensen2006}.  This is
because the particle-based algorithm avoids the need to store and
propagate the error covariance matrix (whose size scales as $d^2$).

An expository review of the continuous-time filters including the
progression from the Kalman filter (1960s) to the ensemble Kalman
filter (1990s) to the feedback particle filter (2010s) appears
in~\cite{TaghvaeiASME2017}.  In continuous-time settings, the first
interacting particle representation of the nonlinear filter appears in
the work of Crisan and Xiong~\cite{crisan10}.  
Also in continuous-time settings, Reich and collaborators have derived
deterministic forms of the EnKF~\cite{reich11,Reich-ensemble}.  
In discrete-time settings, Daum and collaborators have pioneered the
development of closely related particle flow
algorithms~\cite{DaumHuang08,daum2017generalized}.  % Numerical
% comparisons of various particle filter algorithms appears
% in~\cite{stano2014,stano2013nonlinear,berntorp2015,surace_SIAM_Review,adamtilton_fusion13}.  

The technical approach of this paper has its roots
in the optimal transportation theory.  These methods have been widely
applied for uncertainty propagation and Bayesian inference: The
ensemble transform particle filter is based upon computing an optimal
coupling by solving a linear optimization
problem~\cite{reich2013nonparametric}; Polynomial approximations of
the Rosenblatt transport maps for Bayesian inference appears
in~\cite{MarzoukBayesian,mesa2019distributed}; Solution of such
problems using the 
Gibbs flow is the subject of~\cite{heng2015gibbs}.  The time stepping
procedure of this paper is inspired by the J-K-O construction
in~\cite{jordan1998variational}.  Its extension to the filtering
problem appears
in~\cite{laugesen15,halder2017gradient,halder2018gradient,
	halderproximal}.

Closely related to error analysis of this paper is the recent
literature on stability and convergence of the EnKF algorithm.  For
the discrete-time EnKF algorithm, these results appear
in~\cite{gland2009,mandel2011convergence,tong2016nonlinear,stuart2014stability,kwiatkowski2015convergence}.  
The analysis for continuous-time EnKF is more recent. 
For continuous-time EnKF with perturbed observation, under additional
assumptions (stable and fully observable), it has been shown that the
empirical distribution of the ensemble converges to the mean-field
distribution uniformly for all time with the rate
$O(\frac{1}{\sqrt{N}})$~\cite{delmoral2016stability}. This result has
been extended to the nonlinear setting for the case with Langevin type
dynamics with a strongly convex potential and full linear
observation~\cite{delmoral2017stability}. The stability assumption is
recently relaxed in~\cite{bishop2018stability}. Under certain
conditions, convergence and long term stability results appear in~\cite{jana2016stability}. 

In independent numerical evaluations and comparisons, it has been
observed that EnKF and FPF exhibit smaller simulation
variance and better scaling properties -- with respect to the problem
dimension -- when compared to the traditional
methods~\cite{stano2014,stano2013nonlinear,berntorp2015,surace_SIAM_Review,adamtilton_fusion13}.   The error
analysis (in~\Sec{sec:mean_covariance}) together with the
analytical bounds on comparison with the importance sampling
(in~\Sec{sec:PF}) provide the first such rigorous justification for
the performance improvement reported in literature.  The analysis of
this paper is likely to spur wider adoption of the
control-based algorithms for the purposes of sampling and simulation.    

\subsection{Paper outline}
%The outline of the remainder of this paper is as follows:
\Sec{sec:non-uniqueness} includes the preliminaries along
with a discussion of the non-uniqueness issue.  Its resolution is
provided in~\Sec{sec:Opt_FPF} where the optimal FPF is
derived.  The error analysis appears in~\Sec{sec:mean_covariance} and
comparison with importance sampling particle filter is given in~\Sec{sec:PF}.  The
proofs appear in the Appendix. 

\medskip

 \newP{Notation}
%The space of positive symmetric definite matrices of size $d\times d$ is denoted by $S^d_{++}$. 
For a vector $m$, $\|m\|_2$
denotes the Euclidean norm.  For a square matrix $\Sigma$,
$\Fnorm{\Sigma}$
%$\Fnorm{\Sigma}:=\trace(\Sigma\Sigma^\top)$ 
denotes the Frobenius norm,
%$\snorm{\Sigma}:={\max}\{|\Sigma v|;|v|=1\}$ 
$\|\Sigma\|_2$ 
is the spectral norm,
% $\Sigma^\top$ is the matrix-transpose,
$\trace(\Sigma)$ is the
matrix-trace, $\text{Ker}(\Sigma)$ denotes the null-space, $\text{Range}(\Sigma)$ denotes the range space, and $\text{Spec}(\Sigma)$ denotes the spectrum. For a symmetric matrix $\Sigma$, $\lambda_{\text{max}}(\Sigma)$ and $\lambda_{\text{min}}(\Sigma)$ denote the maximum and minimum eigenvalues of $\Sigma$ respectively. The partial order of positive definite matrices is denoted by $\succ$ such that $A \succ B$ means $A-B$ is positive definite.
$\NN(m,\Sigma)$ is a Gaussian probability distribution with mean
$m$ and covariance $\Sigma$.

\section{The Non-Uniqueness Issue}
\label{sec:non-uniqueness}
%\subsection{Linear Gaussian filtering problem} 
\subsection{Preliminaries}
The linear Gaussian filtering problem is described by the linear
stochastic differential equations (sde-s):
\begin{subequations}
\begin{align}
\ud \X_t &= A \X_t \ud t + \sigma_B\ud B_t, %\quad X_0 \sim \NGs(\hat{X}_0,\Sigma_0)
\label{eqn:Signal_Process}
\\
\ud Z_t &= H \X_t\ud t + \ud W_t,
\label{eqn:Obs_Process}
\end{align}
\end{subequations}
where $\X_t\in\Re^d$ and $Z_t \in\Re^m$ are the state and observation
at time $t$, $\{B_t\}_{t\ge 0},\{W_t\}_{t\ge 0}$ are mutually
independent standard Wiener
processes taking values in $\Re^{q}$ and $\Re^m$, respectively, and $A$, $H$, $\sigma_B$ are matrices of appropriate dimension.
% Without loss of generality, it is assumed that the covariance
% matrices of $B_t$ and $W_t$ are identity matrices. 
The initial condition $X_0$ is assumed to have a Gaussian distribution
$\NN(m_0,\Sigma_0)$. 
%with $\Sigma_0\in S^d_{++}$. 
%The initial condition and noise processes are also independent.
The filtering problem is to compute the posterior distribution, 
\begin{equation}\label{eq:posterior}
\pi_t(\cdot):=\PP(\X_t \in \cdot |\clZ_t),
\end{equation}
where $\clZ_t:=\sigma(Z_s; \; 0 \leq s \leq t)$.
%We furthur assume the pair $(A,C)$ is observable. 

%\pgm{Assumption A1 is not necessary here!}
% The following is assumed throughout the remainder of this paper:

% \newP{Assumption (A1)} The system $(A,H)$ is detectable and
% $(A,\sigma_B)$ is stabilizable.    

\medskip

\newP{Kalman-Bucy filter} 
In this linear Gaussian case, the posterior distribution
$\pi_t$ is Gaussian $\NN(m_t,\Sigma_t)$, whose mean
$m_t$ and variance $\Sigma_t$ evolve according to the Kalman-Bucy filter~\cite{kalman-bucy}:
\begin{subequations}
\begin{align}
\ud m_t &= A m_t \ud t + \k_t(\ud Z_t - Hm_t \ud t),
\label{eq:kalman-mean} \\
\frac{\ud}{\ud t} \Sigma_t &= \Ricc(\Sigma_t) \eqdef A\Sigma_t + \Sigma_t A^\top + \Sigma_B - \Sigma H^\top H \Sigma_t,
\label{eq:kalman-variance}
\end{align}
\end{subequations}
where $\k_t \eqdef\Sigma_tH^\top$ is the Kalman gain and $\Sigma_B \eqdef \sigma_B\sigma_B^\top$.

\medskip

%The objective of the interactive particle filtering algorithms is to design a sde for the mean-field process $\{\bar{X}_t\}$ such that the conditional distributions  
\newP{Feedback particle filter} 
The stochastic linear FPF~\cite[Eq. (26)]{yang2016} (and also the square-root form of the EnKBF~\cite[Eq (3.3)]{Reich-ensemble}) is described by the Mckean-Vlasov sde:
\begin{equation}
\begin{aligned}
\ud \Xbar_t= A\Xbar_t \ud t + \sigma_B\ud \bar{B}_t + 
 %R^{-1}
 \underbrace{\Kbar_t \big( \ud Z_t - \frac{H \Xbar_t + H \mbar_t}{2} \ud t
   \big)}_{\text{FPF control law}}, %\\\quad \Xbar_0 \sim \NN(\hat{X}_0,\Sigma_0)
\end{aligned}
\label{eq:FPF-lin}
\end{equation}
where $\Kbar_t \eqdef \Sigmabar_tH^\top$ is the Kalman gain, $\bar{B}_t$ is a standard Wiener process, $\mbar_t \eqdef \E[ \Xbar_t |\clZ_t]$, $\Sigmabar_t
\eqdef \E[ (\Xbar_t - \mbar_t)(\Xbar_t - \mbar_t)^\top |\clZ_t]$ are the mean-field terms, 
%where $\clZ_t := \sigma(Z_s:  s \le t)$. 
and $\Xbar_0\sim\NN(m_0,\Sigma_0)$.  We use 
\begin{equation}\label{eq:posterior-mean-field}
\bar{\pi}_t(\cdot):=\P(\Xbar_t \in \cdot |\clZ_t )
\end{equation} to denote the conditional distribution of mean-field process $\Xbar_t$.  % According to the following Theorem, the mean-field process $\Xbar_t$ is exact, i.e the distribution  of $\bar{X}^i_t$ is equal to the posterior distribution. 

\medskip

The FPF control law is exact.  The exactness result appears in the
following theorem which is a special case of 
% following 
% The following theorem contains the result that the conditional probability distribution of  the mean-field process is equal to the exact posterior distribution given by Kalman filter. 
% It is a special case of 
the~\cite[Thm. 1]{yang2016} that describes the exactness result for
the general nonlinear non-Gaussian case.   
%\pgm{Just refer this theorem to Tao's paper.  Remove all the proofs!}
A proof is included in Appendix~\ref{apdx:exactness-proof}.  The proof
is useful for studying the non-uniqueness issue described in Sec.~\ref{sec:non-uniqueness-subsec}.  
\medskip
\begin{theorem}
\label{thm:exactness} {\bf (Exactness of linear FPF)} Consider the linear
Gaussian filtering
problem~\eqref{eqn:Signal_Process}-\eqref{eqn:Obs_Process} and the linear
FPF~\eqref{eq:FPF-lin}. If $\pi_0 = \bar{\pi}_0$ then
\begin{equation*}
\pi_t = \bar{\pi}_t,\quad \forall t \geq 0.
%\label{eq:exactness-lin}
\end{equation*}
\end{theorem}

\medskip

The notation nomenclature is tabulated in Table~\ref{tab:symbols-states}. 
 
 \begin{table}[t]
 \centering
 \begin{tabular}{|c|c|c|}
 \hline
 Variable & Notation & Equation \\ \hline\hline 
 State of the hidden process & $X_t$ & Eq.~\eqref{eqn:Signal_Process}\\
% \hline
% \vspace*{-0.1in}\\
State of the $i^{\text{th}}$ particle in finite-$N$ sys.& $X_t^i$
&Eq.~\eqref{eq:Xit-s},~\eqref{eq:Xit-d} 
 \\
State of the mean-field model  & $\bar{X}_t$ & Eq.~\eqref{eq:FPF-lin},~\eqref{eq:opt-sde} \\ 
\hline
  \vspace*{-0.1in}\\
  Kalman filter mean and covariance & ${m}_t,{\Sigma}_t$ & Eq.~\eqref{eq:kalman-mean}-\eqref{eq:kalman-variance} 
   \\
 Empirical mean and covariance & $\mN_t,\SigN_t$ & Eq.~\eqref{eq:empr_app_mean_var}
 \\
 Mean-field mean and covariance & $\bar{m}_t,\bar{\Sigma}_t$ & Eq.~\eqref{eq:FPF-lin}-\eqref{eq:opt-sde} 
 \\ \hline
   Conditional distribution of $X_t$ & $\pi_t$ & Eq.~\eqref{eq:posterior} 
 \\
  Conditional distribution of $\Xbar_t$  & $\bar{\pi}_t$ & Eq.~\eqref{eq:posterior-mean-field}
 \\
 Empirical distribution of particles $\{X_t^i\}$ & $\pi^{(N)}_t$ & Eq.~\eqref{eq:posterior-particles} \\
 \hline 
 \end{tabular}
 \caption{Nomenclature.}  
 \label{tab:symbols-states}
 \end{table}

%\begin{remark}
%\cite[Thm. 1]{yang2016} contains the result that FPF is exact for a general filtering problem. However, Theorem~\ref{thm:exactness} does not directly follow from~\cite[Thm. 1]{yang2016} because the assumptions     The proof appears in Appendix~\ref{apdx:exactness-proof} for completeness and its relation to Sec.~\ref{sec:consistency-nonuniqueness}.  
%\end{remark}
\subsection{The non-uniqueness issue}\label{sec:non-uniqueness-subsec}

In the proof of Thm.~\ref{thm:exactness} (given in
Appendix~\ref{apdx:exactness-proof}), it is shown that (i) the
conditional mean
process $\{\bar{m}_t\}_{t\ge 0}$ evolves according
to~\eqref{eq:kalman-mean}; and (ii) the
conditional variance
process $\{\bar{\Sigma}_t\}_{t\ge 0}$ evolves according
to~\eqref{eq:kalman-variance}.

Define an error process $\xi_t := \Xbar_t -
\mbar_t$ for $t\ge 0$.  The equation for $\xi_t$ is obtained by 
subtracting the equation for the mean,~\eqref{eq:Xii-mean} in
Appendix~\ref{apdx:exactness-proof}, from~\eqref{eq:FPF-lin}:
\begin{equation*}
\ud \xi_t = (A - \frac{1}{2} \Sigmabar_t H^TH)\xi_t + \sigma_B \ud \bar{B}_t
\end{equation*}
This is a linear system and therefore, the variance of $\xi_t$, easily seen to be given by $\bar{\Sigma}_t$,
evolves according to the Lyapunov equation
\begin{align*}
\frac{\ud}{\ud t} \Sigmabar_t &=  (A - \frac{1}{2} \Sigmabar_t H^TH) \Sigmabar_t + \Sigmabar_t(A - \frac{1}{2}\Sigmabar_t H^TH) ^\top + \Sigma_B \\&= \Ricc(\Sigmabar_t)
%A \Sigmabar_t + \Sigmabar_tA^T + \Sigma_B - \Sigmabar_tH^TH\Sigmabar_t
%\label{eq:Xii-variance}
\end{align*}
which is identical to \eqref{eq:kalman-variance}. 

These arguments suggest the following general procedure to
construct an exact $\Xbar_t$ process:
Express $\Xbar_t$ as a sum of two terms: \[ \Xbar_t =\mbar_t + \xi_t,\] 
where $\mbar_t$ evolves according to the Kalman-Bucy 
equation~\eqref{eq:kalman-mean} and the evolution of $\xi_t$ is defined by the sde:
\begin{equation*}
\ud \xi_t = G_t \xi_t \ud t + \sigma_t \ud 
\bar{B}_t   +\sigma'_t \ud \bar{W}_t\label{eq:EG}
\end{equation*} 
%Amir: do you want to add copy of W too?
%
where $\{\bar{W}\}_{t\geq 0}$ and $\{\bar{B}\}_{t\geq 0}$ are independent Brownian motions, and  $G_t$, $\sigma_t$, and $\sigma'_t$ are solutions to the matrix equation
\begin{equation}
G_t \Sigmabar_t + \Sigmabar_t G_t^T + \sigma_t \sigma_t^\top  + \sigma'_t( \sigma'_t)^\top= \Ricc(\Sigmabar_t)
\label{eq:G-gen}
\end{equation}
By construction, the equation for the variance is given by the Riccati equation~\eqref{eq:kalman-variance}.

In general, there are infinitely many solutions for \eqref{eq:G-gen}. Below, we describe three solutions that lead to three established form of EnKF and linear FPF:
\begin{enumerate}
	\item EnKF with perturbed observation~\cite[Eq. (27)]{reich11}: 
	\begin{align*}
	G_t = A - \Sigmabar_t H^\top H,\quad \sigma_t = 
	\Sigmabar_t H^\top,\quad \sigma'_t= \sigma_B
	\end{align*}
	\item Stochastic linear FPF~\cite[Eq. (26)]{yang2016} or square-root form of the EnKF~\cite[Eq (3.3)]{Reich-ensemble} : 
	\begin{align*}
	G_t = A - \frac{1}{2}\Sigmabar_t H^\top H,\quad \sigma_t = 
      \sigma_B,\quad \sigma'_t=0
	\end{align*}
	\item Deterministic linear FPF~\cite{AmirACC2018}: 
	\begin{align*}
	G_t = A - \frac{1}{2}\Sigmabar_t H^\top H + \frac{1}{2}\Sigmabar_t^{-1}\Sigma_B,\quad \sigma_t = 0,\quad \sigma'_t=0
	\end{align*}
\end{enumerate}

Given a particular solution $G_t$, one can construct a
family of solutions $G_t + \Sigmabar_t^{-1} \Omega_t$, where
$\Omega_t$ is {\em any} skew-symmetric matrix. 
For the linear Gaussian problem, the non-uniqueness issue has been discussed in literature:  At least
two forms of EnKF, the perturbed observation form~\cite{reich11} and the
square-root form~\cite{Reich-ensemble}, are well known.  
A homotopy of exact deterministic and stochastic filters is given in~\cite{kim2018derivation}. An explanation for the non-uniqueness in
terms of the Gauge transformation appears in~\cite{abedi2019gauge}.

\subsection{Finite-$N$ implementation}
  	In a numerical implementation, one simulates $N$
  	stochastic processes (particles) $\{X_t^i:1\le i\le N\}_{t\ge 0}$, where
  	$X_t^i$ is the state of the
  	$i^{\text{th}}$-particle at time $t$.  The evolution of $X_t^i$ is
  	obtained upon empirically approximating the mean-field terms.  The
  	finite-$N$ filter for the linear FPF~\eqref{eq:FPF-lin} is an
  	interacting particle system:
  	\begin{align}
  	\ud X^i_t &= A X^i_t\ud t + \sigma_B \ud B_t^i + \k^{(N)}_t(\ud Z_t -
  	\frac{HX^i_t + Hm^{(N)}_t}{2}\ud t)
  	\label{eq:Xit-s}
  	\end{align}
  	where $\k^{(N)}_t :=\Sigma^{(N)}_tH^\top$; $\{B^i_t\}_{i=1}^N$ are
  	independent copies of $B_t$; $X^i_0
  	\stackrel {\text{i.i.d}}{\sim} \mathcal{N}(m_0,\Sigma_0)$ for $i=1,2,\ldots,N$; and the
  	empirical approximations of the two mean-field terms are as follows:
  	\begin{equation}
  	\begin{aligned}
  	m^{(N)}_t&:=\frac{1}{N}\sum_{j=1}^N X^i_t\\
  	\Sigma^{(N)}_t
  	&:=\frac{1}{N-1}\sum_{j=1}^N (X^i_t-m^{(N)}_t)(X^i_t-m^{(N)}_t)^\top
  	\end{aligned}
  	\label{eq:empr_app_mean_var}
  	\end{equation}
  	We use the notation
  	\begin{equation}\label{eq:posterior-particles}
  	\pi^{(N)}_t := \frac{1}{N}\sum_{i=1}^N \delta_{X^i_t}
  	\end{equation}
  	to denote the empirical distribution of the particles.  Here,
        $\delta_x$ denotes the Dirac delta distribution at $x$.  
  
%\begin{remark}\label{rem:singular-covariance}
%AMIR: Discuss the empirical covariance being singular issue.  
%\end{remark}

\section{Optimal Transport FPF}
\label{sec:Opt_FPF}

The problem is to uniquely identify an exact stochastic process
$\Xbar_{t}$.  The solution is based on an optimality concept motivated by the optimal
transportation theory~\cite{evans,villani2003}.  The background on
this theory is briefly reviewed next.   

\subsection{Background on optimal transportation}
 \label{apdx:opt-transp}
 Let $\mu_X$ and $\mu_Y$ be two probability measures on $\Re^d$ with finite second moments.
 %Let $X$ be a random variable with probability law $P_X$.
 The (Monge) optimal transportation problem with quadratic cost is to minimize
 \begin{equation}
 \begin{aligned}
 \min_{T}~ \expect[ (T(X)-X)^2 ]%\text{ s.t }X\sim P_X,~T(X)\sim P_Y
 \end{aligned}
 \label{eq:opt-transp}
 \end{equation}   
 over all measurable maps $T:\Re^d \to \Re^d$ such that 
 \begin{equation}
 X\sim \mu_X, ~T(X) \sim \mu_Y
 \label{eq:constraint}
 \end{equation}
 %$ has probability law equal to $P_Y$.
 %This is the optimal transportation problem with quadratic cost. 
 If it exists, the minimizer $T$ is called the {\em optimal transport map} between $\mu_X$ and $\mu_Y$. 
The optimal cost is referred to as $L^2$-Wasserstein distance between $\mu_X$ and $\mu_Y$\cite{villani2003}.
 %\medskip
 \begin{theorem}{(Optimal map between Gaussians~\cite[Prop.~7]{givens84})}  
 	\label{thm:opt-map-Gaussians}
 Consider the optimization problem \eqref{eq:opt-transp}-\eqref{eq:constraint}. Suppose $\mu_X$ and $\mu_Y$ are Gaussian distributions, $\NN(m_X,\Sigma_X)$ and $\NN(m_Y,\Sigma_Y)$, with $\Sigma_X,\Sigma_Y\succ 0$.
 Then the optimal transport map between $\mu_X$ and $\mu_Y$ is given by
 \begin{equation}
 T(x) = m_Y + F(x-m_X)
 \label{eq:optmapgauss}
 \end{equation}
 where 
 %	\begin{equation}
 $
 F = \Sigma_Y^{\half}(\Sigma_Y^{\half}\Sigma_X\Sigma_Y^\half)^{-\half}\Sigma_Y^\half
 $.
\end{theorem}

\subsection{The time-stepping optimization procedure} 
\label{subsec:stepping}

To uniquely identify an exact stochastic process $\Xbar_{t}$, the following time stepping
optimization procedure is proposed:

\begin{enumerate}
\item Divide the time interval $[0,T]$ into $n \in \mathbb{N}$ equal time steps with the time instants $t_0=0 < t_1 < \ldots < t_n = T$. % where $t_{\ii+1}-t_\ii=\Delta t$ for $\ii=1,\ldots,n$. 
%\medskip
\item Initialize a discrete time random process $\{\Xbar_{t_k}\}_{k=1}^n$ according to the initial distribution (prior) of $X_0$,
\begin{equation*}
\Xbar_{t_0} \sim \pi_0
\end{equation*}
\item For each time step $[t_k,t_{k+1}]$, evolve the process $\Xbar_{t_k}$ according to 
\begin{equation}
\Xbar_{t_{\ii+1}}=T_\ii(\Xbar_{t_\ii}),\quad \text{for}\quad \ii=0,\ldots,n-1
\label{eq:updatestep}
\end{equation}
where the map $T_k$ is the optimal transport map between the
probability measures $\pi_{t_k}$ and
$\pi_{t_{k+1}}$. 
%The definition of the optimal transport map is included in Appendix~\ref{apdx:opt-transp}. 

%\medskip
\item Take the limit as $n\to \infty$ to obtain the continuous-time process $\Xbar_t$ and the sde:
\begin{equation}
\ud \Xbar_t = u_t(\Xbar_t)\ud t + \k_t(\Xbar_t)\ud Z_t
\label{eq:opt-sde-uk}
\end{equation} %, where $\Delta t = \underset{\ii}{\sup} ~\{t_{\ii+1}-t_\ii\}$.
%%\medskip
%\item Replace $\PP(\X_t|\clZ_t)$ with $\PP(\Xbar_t|\clZ_t)$. % wherever needed.
\end{enumerate}

%\medskip
The procedure leads to the control laws $u_t$ and $\k_t$ that depend
upon $\pi_t$.  Since $\pi_t$ is unknown, one simply replaces it with
$\bar{\pi}_t$ -- the two are meant to be identical by construction. 
% In practice, one would also need to approximate $\PP(\Xbar_t|\clZ_t)$ by simulating $N$ realizations of the process $\Xbar_t$ (see Remark \ref{remark:finite-N}). 
%In general $u^*$ and $\k^*$ depend on $\PP(\X_t|\clZ_t)$.
The resulting sde~\eqref{eq:opt-sde-uk} is referred to as the optimal
transport FPF or simply the optimal FPF.  A definition is needed
to state the main result.

\begin{definition}
	For  a given positive-definite matrix $Q\succ 0$, define
        $\sRicc(Q)$ as the (unique such) symmetric solution to the matrix equation:
	\begin{equation}
	\sRicc(Q)Q + Q\sRicc(Q) = \Ricc(Q) 
	\label{eq:sRicc-def}
	\end{equation} 
\end{definition}

\begin{remark}
	The symmetric solution to the matrix
	equation~\eqref{eq:sRicc-def} is explicitly given by:
	\begin{equation*}
\sRicc(Q) = \int_0^\infty e^{-sQ}\Ricc(Q)e^{-sQ}\ud s
	\end{equation*}
	The solution can also be expressed as:
	\begin{equation}\label{eq:G-def-with-Omega}
	\sRicc(Q) = A - \frac{1}{2}Q H^\top H + \frac{1}{2}\Sigma_BQ^{-1} + \Omega Q^{-1}
	\end{equation}
	where $\Omega$ is the (unique such) skew-symmetric matrix that solves the matrix equation
	\begin{equation}
	\begin{aligned}
	%	G_t \Sigmabar_t + \Sigmabar_t G_t = \Ricc(\Sigmabar_t)
	\Omega Q^{-1} + Q^{-1}\Omega =
	(A^\top & - A)  
	+ 
	\half(QH^\top H - H^\top HQ) \\+  & \;\;\half (\Sigma_B Q^{-1} -Q^{-1}\Sigma_B)
	\end{aligned}
	\label{eq:Omega}
	\end{equation}
\end{remark}

The main result is as follows.  Its proof appears in the Appendix~\ref{proof:opt-sde}.

\medskip

%\subsection{Main Result}
\begin{proposition}
Consider the linear Gaussian filtering problem~\eqref{eqn:Signal_Process}-\eqref{eqn:Obs_Process}. 
Assume the initial covariance $\Sigma_0 \succ 0$.  
% and Assumption A1 holds. AMIR?
Then the optimal transport FPF is given by
\begin{equation}
\begin{aligned}
\ud \bar{X}_t = 
&A \bar{m}_t\ud t +
\bar{\k}_t(\ud Z_t - H{\bar{m_t}}\ud t) + G_t(\bar{X}_t-\bar{m}_t)\ud t
\end{aligned}
\label{eq:opt-sde}
\end{equation}
where $\Kbar_t:=\Sigmabar_tH^T$, $\mbar_t=\expect [\Xbar_t|\clZ_t]$, $\bar{\Sigma}_t=\expect [(\Xbar_t-\mbar_t)(\Xbar_t-\mbar_t)^\top|\clZ_t]$, $\bar{X}_0 \sim N(m_0,\Sigma_0)$,
and $G_t = \sRicc(\Sigmabar_t)$

The filter is exact: That is, the conditional distribution of
$\bar{X}_t$ is Gaussian $N(\bar{m}_t,\bar{\Sigma}_t)$  with
$\bar{m}_t=m_t$ and $\bar{\Sigma}_t=\Sigma_t$. 

\label{prop:opt-sde-vector}
\end{proposition}

\medskip

%\begin{remark}\normalfont
%	\label{remark:opt-transp-compare}

Using  the form of the solution~\eqref{eq:G-def-with-Omega} for $G_t=\sRicc(\Sigmabar_t)$, the optimal transport sde~\eqref{eq:opt-sde} is expressed as
	\begin{equation}
	\begin{aligned}
	\ud \bar{X}_t = &A \bar{X}_t \ud t + \frac{1}{2}\Sigma_B\Sigmabar_t^{-1}\ (\bar{X}_t - \bar{m}_t)\ud t \\
	&+\frac{1}{2}\bar{\k}_t(\ud Z_t - \frac{H \bar{X}_t+H\bar{m}_t}{2}\ud t) + \Omega_t\Sigmabar_t^{-1}(\bar{X}_t - \bar{m}_t)\ud t 
	\end{aligned}\label{eq:opt-sde-Omega}
\end{equation}
%\end{remark}
%\begin{remark}\normalfont
%\label{remark:opt-transp-compare}
Compared to the original (linear Gaussian) FPF~\eqref{eq:FPF-lin}, the
optimal transport FPF~\eqref{eq:opt-sde-Omega} has two differences:
\begin{enumerate}
\item The stochastic term $\sigma_B\ud \bar{B}_t$ is replaced with the deterministic term $\frac{1}{2}\Sigma_B \bar{\Sigma}_t^{-1}(\bar{X}_t-\bar{m}_t)\ud t$. Given a Gaussian prior, the two terms yield the same posterior.
However, in a finite-$N$ implementation, the difference becomes
significant. The stochastic term serves to introduce an additional error of order $O(\frac{1}{\sqrt{N}})$.  % The error analysis of the deterministic case and the stochastic case is the subject of Sec.~\ref{sec:error-opt-transp-FPF} and Sec.~\ref{sec:error-stoch-FPF} respectively.
\item The sde~\eqref{eq:opt-sde-Omega} has an extra term involving the
  skew-symmetric matrix $\Omega_t$. The extra term does not effect the
  posterior distribution. This term is viewed as a correction term
  that serves to make the dynamics symmetric and hence optimal in the
  optimal transportation sense.  It is noted that for the scalar
  ($d=1$) case, the skew-symmetric term is zero.  Therefore, in the
  scalar case, the update formula in the original FPF~\eqref{eq:FPF-lin} is optimal.  In
  the vector case, it is optimal iff $\Omega_t\equiv 0$.  
% The error anali
%where strong results are obtained. The is equal in distribution. However 
\end{enumerate} 
%\end{remark}

\medskip
\subsection{Finite-$N$ implementation in non-singular covariance case}
 \label{rem:singular-covariance}
The finite-$N$ implementation of the optimal transport sde~\eqref{eq:opt-sde} is given by the following system of $N$ equations:
% given by the sde:
 \begin{equation}
 \begin{aligned}
 \ud {X}_t^i = & A\mN_t \ud t + 
 \K^{(N)}_t (\ud Z_t - H m^{(N)}_t \ud t) \\+ &\sRicc(\SigN_t)({X}_t^i - m^{(N)}_t
 )\ud t
 \end{aligned}
 \label{eq:Xit-d}
 \end{equation}
for $i=1,\ldots,N$, where  $\K^{(N)}_t :=\Sigma^{(N)}_tH^\top$; $X^i_0
\stackrel {\text{i.i.d}}{\sim} \mathcal{N}(m_0,\Sigma_0)$; and empirical
approximations of mean and variance are defined
in~\eqref{eq:empr_app_mean_var}.

The matrix $\sRicc(\SigN_t)$ is not well-defined if $\SigN_t$ is a singular matrix.  This
is a problem because in a finite-$N$ implementation, it is possible
that $\SigN_t \nsucc 0$ even though $\Sigma_t\succ 0$. 
In particular,  when $N<d$, the empirical covariance matrix is of rank
at most $N$ and hence singular. 
%If the covariance matrix is singular, the solution to the Lyapunov
%equation~\eqref{eq:sRicc-def} may not exist. 
Note that this problem {\em only} arises for the optimal and deterministic forms
of the FPF.  In particular, the
stochastic FPF~\eqref{eq:Xit-s} does not suffer from this
issue. It can be simulated for any choice of
$N$. In~\Sec{sec:singular-case}, we extend the optimal transportation
formulation 
to handle also the singular forms of the covariance matrix.

%\end{remark}

\subsection{The singular covariance case}\label{sec:singular-case}

The derivation of the optimal FPF~\eqref{eq:opt-sde}
crucially relies on the assumption that $\Sigmabar_{0}\succ 0$ which in
turn implies that, in the time-stepping procedure,
$\Sigmabar_{t_k}\succ 0$ for $k=0,1,\ldots,n-1$.  In the proof of~\Proposition{prop:opt-sde-vector}, the assumption is used to derive the
optimal transport map $T_k$.  In general, when the
covariance of Gaussian random variables $\Xbar_{t_k}$ or $\Xbar_{t_{k+1}}$ is
singular, the optimal transport map $T_k$ may not exist.  

In the
singular case, a relaxed form of the optimal transportation problem,
first introduced by Kantorovich, is used to search for optimal
(stochastic) couplings instead of (deterministic) transport
maps~\cite{villani2003}.  The following example helps illuminate the
issue:
 % while additional background on the Kantorovich optimal
% transportation problem can be found in Appendix~\ref{apdx:opt-transp}.    

\begin{example}\label{example1}
Consider Gaussian random variable $X$ and $Y$ with distributions,
$\NN(m_X,\Sigma_X)$ and $\NN(m_Y,\Sigma_Y)$, respectively.
Suppose \[m_X=m_Y=\begin{bmatrix} 0 \\ 0 \end{bmatrix},\quad \Sigma_X = \begin{bmatrix}
1 & 0\\0 & \epsilon 
\end{bmatrix}
,\quad 
\Sigma_Y = \begin{bmatrix}
1 & 0\\0 & 1
\end{bmatrix}
\] where $\epsilon\geq0$ is small. 
If $\epsilon >0$, the optimal transportation map exists, and is given by \[Y =  \begin{bmatrix}
1 & 0\\0 & \frac{1}{\sqrt{\epsilon}}
\end{bmatrix}X\] 
If $\epsilon=0$ then there is no transport map that satisfies the
constraints of the optimal transportation problem.

The Kantorovich relaxation of the optimal transportation
problem~\eqref{eq:opt-transp} is the following optimization problem:
\begin{equation}
 \min_{\mu}~\Expect_{(X,Y)\sim \mu}[|X-Y|^2]\label{eq:Kantorovich}
 \end{equation}
 where $\mu$ is a joint distribution on $\Re^d\times \Re^d$ with
 marginals fixed to $\mu_X$ and $\mu_Y$. 
 
Although a deterministic map does not exist for the $\epsilon=0$ problem, a
(stochastic) coupling that solves the Kantorovich
problem~\eqref{eq:Kantorovich} exists and is given by \[Y = X + \begin{bmatrix}
0 \\ 1
\end{bmatrix} B\] where $B \sim \mathcal{N}(0,1)$ is independent of $X$. 
\end{example}

\medskip

In Appendix~\ref{apdx:singular-case}, the Kantorovich relaxation is
used to motivate an optimization problem whose solution yields the
% optimal coupling for the linear FPF.  The considerations are shown to yield the
following extension of the optimal FPF:  
%The filter is given by the sde:
\begin{equation}
	\label{eq:opt-sde-singular}
	\ud \Xbar_t = A\mbar_t + \Kbar_t(\ud Z_t - H\mbar_t   \ud t) + G_t(\Xbar_t-\mbar_t)\ud t + \sigma_t\ud \Bbar_t
	\end{equation}
with $\sigma_t := 
P_K\sigma_B$ where
$P_K$ is
the projection matrix into the kernel of the matrix $\bar{\Sigma}_t$, and $G_t$ is any symmetric solution of the matrix equation
\begin{align}
	%	\sigma^*_t &= e_t\\
	G_t \Sigmabar_t + \Sigmabar_t G_t& = \Ricc(\Sigmabar_t) - \sigma_t (\sigma_t)^\top  \label{eq:Lyapunov-G-e}
	%	A - \frac{1}{2}\bar{\Sigma}_t H H^\top  + \frac{1}{2}(\sigma_B + e_t) u_t^\top + \Omega^{(1)}_t
	\end{align}

\begin{remark}
	When $\bar{\Sigma}_t$ is singular, the solution to the matrix
        equation~\eqref{eq:Lyapunov-G-e} is not unique. 
It is shown in Appendix~\ref{apdx:singular-case} that the solution % to
% the matrix equation~\eqref{eq:Lyapunov-G-e} 
is of the following
general form:
        \begin{equation}
	\begin{aligned}
	G_t = &A - \frac{1}{2}\bar{\Sigma}_t H H^\top  +
                \frac{1}{2}P_R \Sigma_B  \Sigmabar_t^+  + P_K \Sigma_B \Sigmabar_t^+ + \Sigmabar_t^+ \Sigma_B P_K  \\&+  P_R\Omega^{(1)}\Sigmabar_t^+ + P_K(\Omega^{(0)}  - A)P_K
%                \\&
%   + P_R\Omega^{(1)}_t \Sigmabar_t^{+} +
%   P_K\Omega^{(0)}_tP_K 
	\end{aligned}
	\label{eq:singular-G-Omega}
\end{equation}
	where $\Sigmabar_t^+$ is the pseudo inverse\footnote{The
          pseudo inverse of matrix $Q$ is the unique matrix $Q^+$ that
          satisfies $Q^+QQ^+=Q^+$, $QQ^+Q=Q$, $Q^+Q$ is symmetric, and
          $QQ^+$ is symmetric~\cite{ben2003generalized}.} of
        $\Sigmabar_t$, $P_R$ is projection matrix onto the range of
        $\Sigmabar_t$, and $\Omega^{(1)} \in \Re^{d\times d}$ is a
        skew-symmetric matrix that solves a certain matrix
        equation~\eqref{eq:singular-Omega} and $\Omega^{(0)}$ is an arbitrary symmetric matrix.

Using the formula~\eqref{eq:singular-G-Omega}, the optimal
FPF~\eqref{eq:opt-sde-singular} is expressed as follows:
	\begin{equation}
	\begin{aligned}
	\ud \Xbar_t = &A\Xbar_t \ud t + \frac{1}{2}(\sigma_B + \sigma_t) \sigma_B^\top \Sigmabar_t^+(\Xbar_t-\mbar_t)\ud t + \sigma_t \ud \Bbar_t \\&+ \Kbar_t(\ud Z_t - \frac{H\Xbar_t + H\mbar_t}{2}\ud t) 
	+ \text{[additional terms]}
	\end{aligned}
	\label{eq:opt-sde-singular-Omega}
	\end{equation}
The formula~\eqref{eq:opt-sde-singular-Omega} allows one to clearly
see the relationship between the deterministic and stochastic forms of
the optimal FPF.  In particular, when the covariance matrix is non-singular,
 $\Sigmabar_t^+ = \Sigmabar_t^{-1}$, 
and
$\sigma_t = P_K\sigma_B=0$.  This results in the deterministic form of
the optimal transport
FPF~\eqref{eq:opt-sde-Omega}. When the covariance matrix is singular, then the effect of the linear term $\frac{1}{2}\Sigma_B \Sigmabar_t^{-1}(\Xbar_t-\mbar_t) \ud t$ in~\eqref{eq:opt-sde-Omega} is now simulated with the two terms $ \frac{1}{2}(\sigma_B + \sigma_t) \sigma_B^\top \Sigmabar_t^+(\Xbar_t-\mbar_t)\ud t + \sigma_t \ud \bar{B}_t$ in~\eqref{eq:opt-sde-singular-Omega}. This is conceptually similar to the Example~\ref{example1}, where the deterministic optimal transport map is replaced with a stochastic coupling.        
%The purpose of the stochastic term $\sigma_t \ud \bar{B}_t$ is to account for the the effect that is not accounted for by the
%linear term $\frac{1}{2}\Sigma_B \Sigmabar_t^{-1}(\Xbar_t-\mbar_t) \ud t$. 
The [additional terms] in~\eqref{eq:opt-sde-singular-Omega} do not
affect the distribution. 
\end{remark} 
% Compared to the optimal transport sde~\eqref{eq:opt-sde}, a stochastic term $e_t \ud \bar{B}_t$ is added to account for the the effect of the process noise $\sigma_B \ud B_t$  that can not be captured by the linear term $G_t \bar{\Xbar}_t \ud t$.

%The optimal transportation procedure does not seem to work
%when the covariance matrix $\Sigmabar_t$ is singular. Because in this case, the optimal transport map may not exist. As a result, one need to look for stochastic maps. Consider the following parameterization of the sde, motivated by the three forms of linear FPF discussed in Section~\ref{sec:consistency-nonuniqueness}. 
%%Decompose $\bar{X}_t = \bar{m}_t + \bar{\Xbar}_t$ where 
%\begin{align*}
%\ud \bar{X}_t = \ud \bar{m}_t + 
%%\ud \bar{m}_t &= A \bar{m}_t + \Kbar_t(\ud Z_t - H\bar{m}_t\ud t) \\
%%\ud \bar{\Xbar}_t &=  G_t\bar{\Xbar}_t\ud t + \sigma_t \ud \bar{B}_t
%\end{align*} 

\subsection{Finite-$N$ implementation in the singular covariance case}
%\pgm{AMIR: Just give the singular case.  And allude to simplification
%  in the case non-singular}
%
%BLAH - Lets chat.
% n a numerical implementation of the optimal transport linear FPF algorithm~\eqref{eq:opt-sde-singular}, one simulates $N$
% particles by empirically approximating the mean-field terms.  
% The evolution of the particles  is given by the sde

The finite-$N$ implementation of the optimal transport
sde~\eqref{eq:opt-sde-singular} is given by the following system of
$N$ equations:
\begin{equation}
\begin{aligned}
\ud X^i_t = &A\mN_t \ud t + 
\k^{(N)}_t (\ud Z_t - H m^{(N)}_t \ud t)  + G^{(N)}_t({X}_t^i - m^{(N)}_t)\ud t+\\& \sigma_t^{(N)} \ud B^i_t,\quad X^i_0 \sim \calN(m_0,\Sigma_0)
\end{aligned}
\label{eq:Xi-d-singular}
\end{equation}
where $\K^{(N)}_t :=\Sigma^{(N)}_tH^\top$; $\{B^i_t\}_{i=1}^N$ are
independent copies of $B_t$, $\sigma_t^{(N)} = \sigma_B -\SigN_t (\SigN_t)^{+}\sigma_B$, $G^{(N)}_t$ is a symmetric matrix solution to the matrix equation 
\begin{align*}
%	\sigma^*_t &= e_t\\
G^{(N)}_t \SigN_t + \SigN_t G^{(N)}_t& = \Ricc(\SigN_t) - \sigma_t^{(N)} \sigma_t^{(N)^\top} 
%	A - \frac{1}{2}\bar{\Sigma}_t H H^\top  + \frac{1}{2}(\sigma_B + e_t) u_t^\top + \Omega^{(1)}_t
\end{align*}
and $\mN_t$ and $\SigN_t$ are empirical mean and covariance  defined
in~\eqref{eq:empr_app_mean_var}. Note that the stochastic term is zero
when $\sigma_B \in \text{Range}(\SigN_t)$, which is true, e.g., when $\sigma_B \in \text{span}\{X^1_t,\ldots,X^N_t\}$. 

%not invertible when the number of particles $N$ is less than the dimension $d$.  Therefore, in this case, it is not possible ti implement the  
\section{Error analysis}
\label{sec:mean_covariance}
 
This section is concerned with error bounds in the large but finite
$N$ regime.  Given that $N$ is large, we restrict ourselves to the
non-singular case.  For the purposes of the error analysis, the
following assumptions are made:

\newP{Assumption (I)} The system $(A,H)$ is detectable and
$(A,\sigma_B)$ is stabilizable.     

\newP{Assumption (II)} % The initial condition $X^i_0$ is sampled
% % according to a distribution with finite BLAH moments. 
% %There exists positive constants $L_1,L_2>0$ such that $\snorm{\SigN_0} \in [L_1,L_2]$. 
%There exists positive constants $\alpha_1,\alpha_2>0$ such that $\alpha_1 I \preceq \Sigma_0^{(N)} \preceq \alpha_2 I$.     
Assume $N>d$ and the initial empirical covariance matrix $\SigN_0
\succ 0$ almost surely. 

\medskip

The main result for the finite-$N$ deterministic FPF~\eqref{eq:Xit-d} is as follows with
the proof given in Appendix~\ref{apdx:mean-var-error}.  

\medskip

\begin{proposition} \label{prop:conv_error}
Consider the Kalman filter~\eqref{eq:kalman-mean}-\eqref{eq:kalman-variance}
initialized with the prior ${\cal N}(m_0,\Sigma_0)$ and the
finite-$N$ deterministic form of the optimal FPF~\eqref{eq:Xit-d} initialized with
$X_0^i\stackrel{\text{i.i.d}}{\sim} {\cal N}(m_0,\Sigma_0)$ for $i=1,2,\ldots,N$.  Under
Assumption (I) and (II), the following characterizes the convergence and
error properties of the empirical mean and covariance
$(m_t^{(N)},\Sigma_t^{(N)})$ obtained from the finite-$N$ filter
relative to the mean and
covariance $(m_t,\Sigma_t)$ obtained from the Kalman filter:
\begin{romannum}
\item{Convergence:} For any finite $N>1$, as $t\rightarrow\infty$:
\begin{align*}
\lim_{t \to \infty} e^{\lambda t}\|\mN_t - m_t\|_2 &= 0\quad \text{a.s}\\
\lim_{t \to \infty} e^{2\lambda t}\Fnorm{\SigN_t - \Sigma_t} &= 0\quad \text{a.s}
\end{align*} 
for all $\lambda\in (0,\lambda_0)$ where $\lambda_0$ is a fixed positive
constant (see~\eqref{eq:lambda0_defn} in Appendix~\ref{apdx:mean-var-error}).  
\item{Mean-squared error:} For any $t>0$, as $N\rightarrow\infty$:
\begin{subequations}
	\begin{align}
	\Expect[\|\mN_t-m_t\|_2^2]&\leq \text{(const.)}e^{-2\lambda t}\frac{\trace(\Sigma_0)+\trace(\Sigma_0)^2}{N} \label{eq:mean-estimate}\\
	\Expect[\Fnorm{\SigN_t-\Sigma_t}^2]&\leq \text{(const.)}e^{-4\lambda t}\frac{\trace(\Sigma_0)^2}{N} \label{eq:var-estimate}
	\end{align} 
\end{subequations}
for all $\lambda \in (0,\lambda_0)$. The constant depends on $\lambda$, and spectral norms $\|\Sigma_0\|$, $\|\Sigma_\infty\|_2$ and $\|H\|_2$, where $\Sigma_\infty$ is the solution to the algebraic Riccati equation (see Lemma~\ref{lem:KF-stability}).  

\end{romannum}

\end{proposition}

\medskip

%\pgm{Amir: Combine the remarks into a single remark with an enumerate list}
\begin{remark}\label{rem:scaling-dim} Two remarks are in order:
	\begin{enumerate}
		\item Asymptotically (as $t \to \infty$) the empirical mean and variance of
		the finite-$N$ filter becomes exact.  This is because of the stability
		of the Kalman filter whereby the filter forgets the initial
		condition. In fact, for any (not necessarily i.i.d Gaussian)
		$\{X_0^i\}_{i=1}^N$, that satisfy the assumption $\SigN_0\succ 0$, the result holds.
		
		\item 	(Scaling with dimension) If the parameters of
                  the linear Gaussian filtering
                  problem~\eqref{eqn:Signal_Process}-\eqref{eqn:Obs_Process}
                  scale with the dimension in a way that the spectral
                  norms $\|\Sigma_0\|_2$, $\|\Sigma_\infty\|_2$,
                  $\|H\|_2$, and $\lambda_0$ do not change, then the
                  constant in the error
                  bounds~\eqref{eq:mean-estimate}-\eqref{eq:var-estimate}
                  do not change. The only term that scales with the
                  dimension is  $\trace(\Sigma_0)$. For example, with
                  $\Sigma_0 = \sigma_0^2 I_{d \times d}$, $\trace(\Sigma_0)= d \sigma_0^2$. Therefore,
                  the bound on the error typically scales as $d^2$ in
                  problem dimension.

%		\item 	The error estimates~\eqref{eq:mean-estimate}-\eqref{eq:var-estimate}
%                  hold more broadly for any skew-symmetric choice of
%                  $\Omega$ in the definition of $\sRicc(\SigN_t)$ in~\eqref{eq:G-def-with-Omega}. Therefore, the optimal  choice of $\Omega$ does not effect the error estimates for mean and variance. 
%				In the next section, we study the error for estimating
%                expectation of an arbitrary function $f$, where the
%                choice of $\Omega$ can make a difference.  
	\end{enumerate}

\end{remark}

\subsection{Propagation of chaos}
\label{sec:poa}

In this section, we study the convergence of the empirical
distribution of the particles $\pi^{(N)}_t$ for the finite-$N$
system~\eqref{eq:Xit-d} to the exact posterior distribution $\pi_t$. 
% The posterior distribution is exactly equal to the distribution of the
% mean-field process~\eqref{eq:opt-sde}.  
Derivation of error estimates involve construction of $N$ independent
copies of the exact process~\eqref{eq:opt-sde}.  Consistent with
the convention to denote mean-field variables with a bar, the
stochastic processes are denoted as $\{\bar{X}_t^i:1\le i \le N\}$
where $\bar{X}^i_t$ denotes the state of the $i^{\text{th}}$ particle
at time $t$.  The particle evolves according to the mean-field
equation~\eqref{eq:opt-sde} as
\begin{align}
\ud \bar{X}^i_t & = A \bar{m}_t \ud t + 
\bar{\k}_t(\ud Z_t - H{\bar{m}_t}\ud t) + \sRicc(\Sigmabar_t)(\bar{X}_t^i -\bar{m}_t)\
 \label{eq:barXit}
\end{align} 
where $\bar{\k}_t =\bar{\Sigma}_tH^\top$ is the Kalman
gain and the initial condition $\bar{X}^i_0=X^i_0$, the right-hand side
being the initial condition of
the $i^{\text{th}}$ particle in the finite-$N$ FPF~\eqref{eq:Xit-d}. The mean-field process $\bar{X}^i_t$ is thus coupled to $X^i_t$
through the initial condition.

 In order to carry out the error analysis,  the following event is
 defined for an arbitrary choice of a fixed matrix $\Lambda_0 \succ 0$:
\begin{equation}
\label{eq:S-def}
S_{\Lambda_0}\eqdef\{\SigN_0 \succ \Lambda_0\}
\end{equation}
%\newP{Assumption (II-B)} There exists a matrix $\Lambda_0\succ 0$ such that $\SigN_0 \succ \Lambda_0$ almost surely. 
The following proposition characterizes the error between $X^i_t$ and
$\bar{X}^i_t$, when the event $S_{\Lambda_0}$ is true (the estimate is key to the propagation of
chaos analysis). The proof appears in the
Appendix~\ref{apdx:prop-chaos}.

\begin{proposition} \label{prop:prop-chaos}
Consider the stochastic processes $X^i_t$ and $\bar{X}^i_t$ whose
evolution is defined 
according to the optimal transport FPF~\eqref{eq:Xit-d} and its mean-field
model~\eqref{eq:barXit}, respectively. The initial condition $X_0^i\stackrel{\text{i.i.d}}{\sim} {\cal N}(m_0,\Sigma_0)$ for $i=1,2,\ldots,N$ 
%and the dimension $d=1$. 
 Then, under
Assumptions~(I) and (II):  
\begin{equation}
\Expect[\|X^i_t-\bar{X}^i_t\|^2_2\mathds{1}_{S_{\Lambda_0}}]^{1/2} \leq
\frac{\text{(const.)}}{{\sqrt{N}}}
\label{eq:estimate_POA}
\end{equation}  

\end{proposition}
\medskip

The estimate~\eqref{eq:estimate_POA} is used to prove the following
important result that the empirical distribution of the particles in
the linear FPF converges weakly to the true posterior distribution.
Its proof appears in the Appendix~\ref{apdx:prop-chaos}.

\begin{corollary} \label{cor:prop-chaos}
Consider the linear filtering problem~\eqref{eqn:Signal_Process}-\eqref{eqn:Obs_Process}
and the finite-$N$ deterministic FPF~\eqref{eq:Xit-d}. The initial condition $X_0^i\stackrel{\text{i.i.d}}{\sim} {\cal N}(m_0,\Sigma_0)$ for $i=1,2,\ldots,N$. Under Assumptions (I) and (II), for any bounded and Lipschitz 
function $f:\Re^d \to \Re$,
\begin{equation*}
\Expect\left[\left|\frac{1}{N}\sum_{i=1}^N
    f(X^i_t)-\Expect[f(X_t)|\clZ_t]\right|^2\right]^{1/2} \leq \frac{\text{(const.)}}{{\sqrt{N}}} 
\end{equation*}
%for a positive constant $c>0$.
\end{corollary}

%
%\begin{remark}
%The error analysis result in this section holds for any choice of skew-symmetric $\Omega$, in the definition~\eqref{eq:G-def-with-Omega}. 
%\end{remark}

%\subsection{Stability of optimal transport FPF}
%The stability in filtering, refers to the property that the filter distribution converges to the correct posterior distribution, in the limit as $t \to \infty$, even with wrong initial prior distribution.   
%To study the stability of optimal transport FPF, use the decomposition  $\Xbar_t = \xi_t + \mbar_t$. The process for the mean $\mbar_t$ is stable because of Prop.~\ref{prop:conv_error}-(i). For the error process, its evolution is given by $\ud \xibar_t = \sRicc(\Sigmabar_t)\xibar_t \ud t$.  This linear system is not stable. For example, if $\Sigmabar_0 = \Sigma_\infty$ where $\Sigma_\infty$ is the solution to the algebraic Riccati equation $\Ricc(\Sigma_\infty)=0$, then $\Sigmabar_t=\Sigma_\infty$ and $\sRicc(\Sigma_t) = \sRicc(\Sigma_\infty)=0$. The solution is $\xibar_t= \xi_0$, and converges to the wrong distribution as $t \to \infty$, if the initial distribution is not Gaussian.   This is in contrast to other forms of EnKF algorithms, EnKF with perturbed observation~\cite[Eq. (27)]{reich11}, and Stochastic linear FPF~\cite[Eq. (26)]{yang2016} or square-root form of the EnKF~\cite[Eq (3.3)]{Reich-ensemble}, where the error process is stable, under certain conditions.  
% 
\section{Comparison to Importance Sampling}
\label{sec:PF}

For the purposes of the comparison of the optimal FPF with the
importance sampling-based particle filter, we restrict to the static
filtering example with a fully observed observation model:
\begin{equation}
\begin{aligned}
\ud X_t &= 0\quad\quad X_0 \sim \mathcal{N}(0,\sigma_0^2 I_d)\\
\ud Z_t &= X_t \ud t +  \sigma_w \ud W_t
\end{aligned}
\label{eq:filter-example}
\end{equation}
for $t\in[0,1]$, where $\sigma_W,\sigma_0>0$. 
The posterior distribution at time $t=1$, denoted as $\pi_{\text{EXACT}}$, is a
Gaussian $\NN(m_1,\Sigma_1)$ with $m_1= \frac{\sigma_0^2}{\sigma_0^2
  + \sigma_W^2}Z_1$ and $\Sigma_1
=\frac{\sigma_0^2\sigma_w^2}{\sigma_0^2 + \sigma_w^2} I_d$.

Let $\{X^i_0\}_{i=1}^N$ be $N$ i.i.d samples of $X_0$.  The importance
sampling-based particle filter yields an empirical approximation of the
posterior distribution $\pi_{\text{EXACT}}$ as follows: 
\begin{equation}\label{eq:PF-Est} 
\pi^{(N)}_\text{PF} = \sum_{i=1}^N w_i\delta_{X^i_0},\quad w_i = \frac{e^{-\frac{\|Z_1-X^i_0\|_2^2}{2\sigma_w^2}}}{\sum_{i=1}^N e^{-\frac{\|Z_1-X^i_0\|_2^2}{2\sigma_w^2}}}
\end{equation}
In contrast, given the initial samples $\{X^i_0\}_{i=1}^N$, the FPF approximates
the posterior by implementing a feedback control law as follows:
\begin{equation}\label{eq:EnKF-est}
\pi^{(N)}_{\text{FPF}} = \frac{1}{N}\sum_{i=1}^N \delta_{X^i_1},\quad \ud X^i_t = \frac{1}{\sigma_w^2}\SigN_t (\ud Z_t - \frac{X^i_t + \mN_t}{2} \ud t) 
\end{equation}
where the empirical mean $\mN_t$ and covariance $\SigN_t$ are
approximated as~\eqref{eq:empr_app_mean_var}.

The m.s.e in estimating the conditional expectation of a given
function $f$ is defined as follows:
\begin{equation*}
\text{m.s.e}_*(f) = \Expect[\|\pi_*^{(N)}(f) - \pi_{\text{EXACT}}(f)\|_2^2]
\end{equation*}
where the subscript $*$ is either the $\text{PF}$ or the
$\text{FPF}$.  

For $f(x) = \frac{1}{\sqrt{d}}1^\top x$, a numerically computed plot
of the level-sets of m.s.e, as a function of $N$ and $d$, is depicted in Figure~\ref{fig:error-PF-FPF}-(a)-(b).
The expectation is approximated by averaging over $M=1000$ independent
simulations.  It is observed that, in order to have the same error,
the importance sampling-based approach requires the number of samples
$N$ to grow exponentially with the dimension $d$, whereas the growth
using the FPF for this numerical example is $O(d^\half)$.  This conclusion is consistent
with other numerical studies reported in the literature~\cite{surace_SIAM_Review}.

% The m.s.e error, for importance sampling and FPF approach, is numerically computed and averaged over $1000$ independent simulations for the function $f(x) = \frac{1}{\sqrt{d}}1^\top x$. The level sets of the numerical m.s.e error as a function of  dimension and sample size are depicted in Figure~\ref{fig:error-PF-FPF}-(a)-(b). It is observed that, in order to keep the same amount of error, the importance sampling approach requires the number of samples to grow exponentially with dimension, whereas the growth for EnKF approach is polynomial. 
\begin{figure*}
	\centering
	\begin{tabular}{cc}
		\subfigure[]{\includegraphics[width = 0.4\hsize]{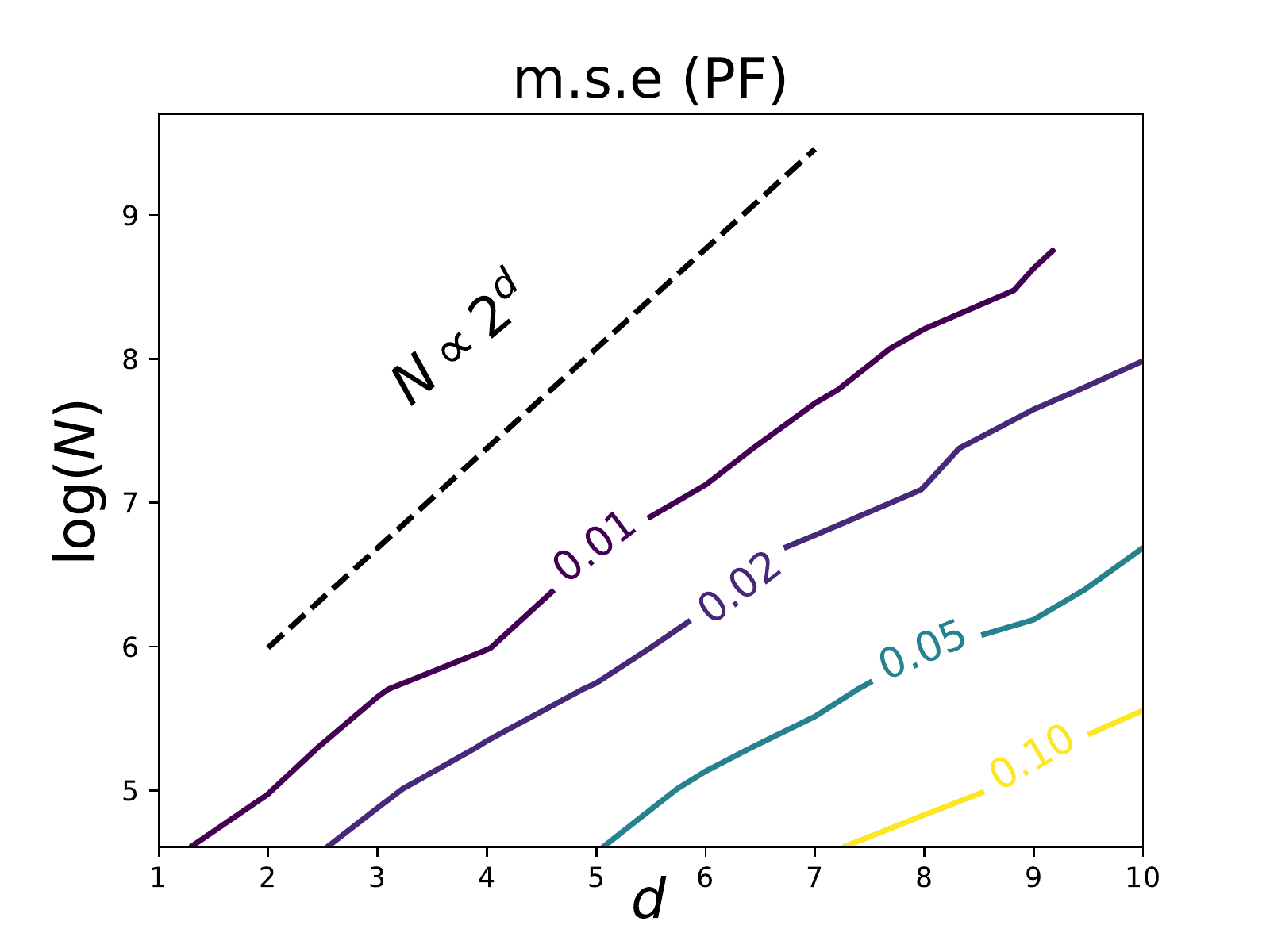}\label{fig:error-PF}} &\subfigure[]{\includegraphics[width = 0.4\hsize]{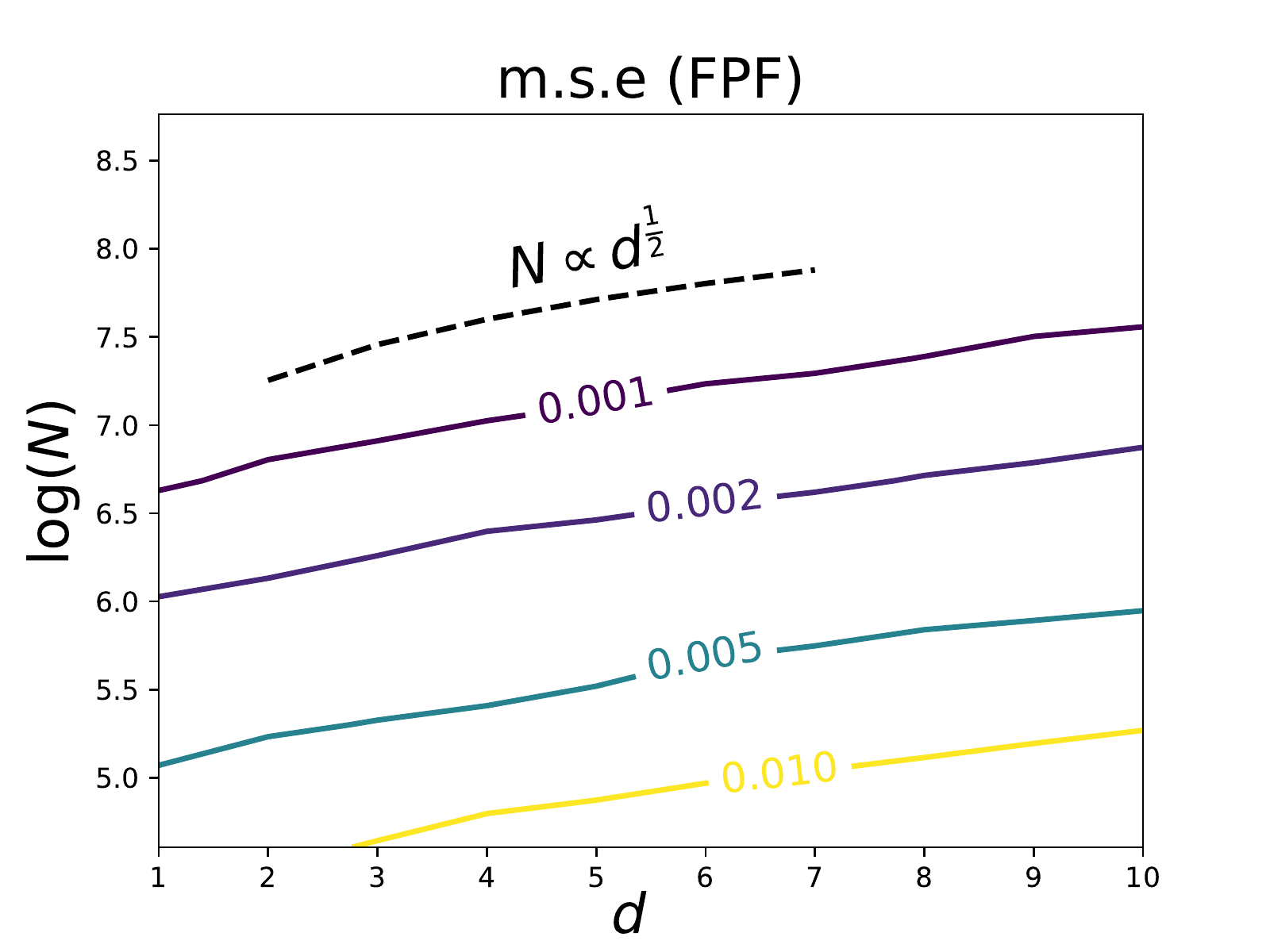}\label{fig:error-FPF}}
	\end{tabular}
	\caption{Level sets of the m.s.e. for the filtering
          problem~\eqref{eq:filter-example}: (a) importance
          sampling algorithm~\eqref{eq:PF-Est} and (b) the FPF~\eqref{eq:EnKF-est} algorithm.} 
	\label{fig:error-PF-FPF}
\end{figure*}

For the purposes of the analysis, a modified form of the particle
filter is considered whereby the denominator is replaced by its exact
form: 
\begin{equation}
\pi^{(N)}_{\overline{\text{PF}}} =  \frac{1}{N} \sum_{i=1}^N \bar{w}_i\delta_{X^i_0},\quad \bar{w}_i = \frac{e^{-\frac{\|Z_1-X^i_0\|_2^2}{2\sigma_w^2}}}{\Expect[e^{-\frac{\|Z_1-X^i_0\|_2^2}{2\sigma_w^2}}|\clZ_1]}
\label{eq:PF-Est-bar}
\end{equation}

The proof of the following result on the scaling of the m.s.e. with
the state dimension $d$ appears in the
Appendix~\ref{apdx:importance-sampling}:

\medskip

\begin{proposition}\label{prop:importance-sampling}
	Consider the filtering problem~\eqref{eq:filter-example} with
        state dimension $d$.  
% , the
        % modified form of the importance sampling
        % estimator~\eqref{eq:PF-Est-bar}, and the FPF
        % estimator~\eqref{eq:EnKF-est}.
        Suppose
        $\sigma_0=\sigma_w=\sigma>0$ and $f(x)=a^\top x$ where
        $a\in\Re^d$ with $\|a\|_2=1$.  Then:
	\begin{romannum}
		\item The m.s.e. for the modified form of the
                  importance sampling estimator~\eqref{eq:PF-Est-bar} is given by 
		\begin{equation}\label{eq:PF-mse-bound}
		\text{m.s.e}_{\overline{\text{PF}}}(f) = \frac{\sigma^2}{N}\left(3(2^d) - \frac{1}{2}\right) \geq \frac{\sigma^2}{N}2^{d+1}
		\end{equation}
		\item The m.s.e for the FPF estimator~\eqref{eq:EnKF-est} is bounded as
		\begin{equation}\label{eq:FPF-mse-bound}
		\text{m.s.e}_{{\text{FPF}}}(f) \leq  \frac{\sigma^2}{N} (3d^2+2d)
		\end{equation}
	\end{romannum}
\end{proposition}

\medskip

\begin{remark}
In the limit as $d\to\infty$, the performance of the importance
sampling-based particle filters has been studied in the
literature~\cite{bickel2008sharp,bengtsson08,snyder2008obstacles,rebeschini2015can}. The
main result in these papers is concerned with the particle degeneracy
(or the weight collapse) issue:  it
is shown that if $\frac{\log N \log d}{d}\to 0$ then the largest
weight $\max_{1\le i \le N} w^i \to 1$ in probability.  Consequently, in order to prevent the
weight collapse,
the number of particles should grow exponentially with the
dimension. This phenomenon is referred to as the curse of
dimensionality for the particle filters. 

\end{remark}

\begin{remark}
The scaling with dimension depicted in
Figure~\ref{fig:error-PF-FPF}~(b) suggests that the $O(d^2)$ bound for
the m.s.e in the linear FPF is loose.  This is the case because, in deriving
the bound, the inequality 
$\|\cdot\|_2\leq \|\cdot\|_F$ was used.  The inequality is loose
particularly so as the dimension grows.  Also, it is observed that the m.s.e for the particle filter grows slightly slower than the lower-bound $2^d$. This is because the lower-bound is obtained for the modified particle filter~\eqref{eq:PF-Est-bar}, while the m.s.e is numerically evaluated for the standard particle filter~\eqref{eq:PF-Est}. The correlation between the numerator and denominator in~\eqref{eq:PF-Est} reduces the m.s.e error. 
\end{remark}

%\begin{remark}
%	The error result regarding the particle filter is expected to hold in discrete-time setting and with resampling. 
%\end{remark}
\section{Conclusions \& Directions for Future Work}

In this paper, a principled approach is presented for design of the EnKF
algorithm.  The approach is based upon a reformulation of the
filtering problem into an optimal transportation problem.  Its
solution is referred to as the optimal transport FPF.  Empirical
approximation of the mean-field terms in the control law yield a
finite-$N$ form of the algorithm as a controlled interacting particle
system.    
Detailed error
analysis is presented for the finite-$N$ algorithm including a
comparison with the importance sampling-based approach.  Taken together with numerical
comparisons in recent literature, the analysis of this paper is likely
to spur research and application of controlled interacting particle
algorithms for filtering and data assimilation.  
There are many directions for future work: 
\begin{enumerate}
\item Apart from the optimal transportation formulation stressed in
  this paper, one may consider alternative approaches for control
  design.  One possible direction is based on the
  Schr\"odinger bridge problem~\cite{chen2016relation,reich2019data}.
  It is an interesting question whether such an approach can lead to
  stochastic forms of FPF, in contrast to the deterministic form
  obtained using the optimal transport formulation.
\item An important research direction is to extend the stability and error
  analysis to the class of finite-$N$ stochastic EnKF and FPF
  algorithms.  The results in this paper serve as baseline, in
  terms of assumptions and order scalings, for the analysis of the
  stochastic algorithm.
\item It will be of interest to construct optimization formulations
  that directly yield a finite-$N$ algorithm without the need for
  empirical approximation as an intermediate step.  Such constructions
  may lead to better error properties by design.
\item Finally, it is extremely important to understand the curse of
  dimensionality (CoD) for general types of controlled interacting particle
  systems.  The result in Prop.~\ref{prop:importance-sampling} is
  very exciting because it suggests that CoD is avoided in the linear
  Gaussian case.  Whether such a property also holds for the
  non-Gaussian case remains an open question.  
\end{enumerate}

\appendices

\section{Proof of Thm.~\ref{thm:exactness}}\label{apdx:exactness-proof}
\begin{proof}
	It is first shown that the conditional mean and variance of $\Xbar_t$ evolve according to Kalman filtering equations. 
	% evolving according to  \eqref{eq:kalman-mean}-\eqref{eq:kalman-variance}. 
	Express the sde~\eqref{eq:FPF-lin} in its integral form,
	\begin{equation*}
	\Xbar_t = \Xbar_0 + \int_0^t A \Xbar_s \ud s + \int_0^t \sigma_B \ud \bar{B}_s + \int_0^t \Kbar_s(\ud Z_s - \frac{H\Xbar_s + H\mbar_s}{2}\ud s)
	\end{equation*}
	Upon taking the conditional expectation of both sides 
	\begin{align*}\mbar_t &= \Expect[\Xbar_0|\clZ_t] +  \E[ \int_0^t A\Xbar_s \ud s|\clZ_t] + \E[\int_0^t \sigma_b \ud \bar{B}_s|\clZ_t]  \\&+\Expect[\int_0^t \Kbar_s(\ud Z_s - \frac{H\Xbar_s + H\mbar_s}{2}\ud s)|\clZ_t]\\&= \Expect[\Xbar_0|\clZ_0] +  \int_0^t \Expect[\Kbar_s|\clZ_s]\ud Z_s\\ &+ \int_0^t\Expect[A\Xbar_s - \Kbar_s\frac{H\Xbar_s + H\mbar_s}{2}|\clZ_s]\ud s
	\\&=\mbar_0 + \int_0^t A\mbar_s \ud s +  \int_0^t \Kbar_s(\ud Z_s - H\mbar_s\ud s)
	\end{align*}
	where we used the fact that $\Xbar_t$ is adapted to the
	filteration $\clZ_t$ to obtain the second identity
	(see~\cite[Lemma 5.4]{xiong2008}). As a result, the sde for
	the conditional mean is 
	% using
	%linearity and the fact that $B_t^i$ is zero mean. 
	\begin{equation}
	\ud \mbar_t = A \mbar_t \ud t + \Kbar_t(\ud Z_t - H\mbar_t\ud t)
	\label{eq:Xii-mean}
	\end{equation}
	Define the error $\xi_t$ according to
	%\begin{equation*}
	$\xi_t := \Xbar_t - \mbar_t$. %\label{eq:decompose}
	%\end{equation*}
	The equation for $\xi_t$ is obtained by simply
	subtracting~\eqref{eq:Xii-mean} from~\eqref{eq:FPF-lin}:
	\begin{equation*}
	\ud \xi_t = (A - \frac{\Sigmabar_t H^TH}{2}) \xi_t + \sigma_B\ud \bar{B}_t \label{eq:E}
	\end{equation*}
	By application of the It\^o's  rule 
	\begin{align*}
	\ud ({\xi}_t{\xi}_t^\top) = &(A -\frac{1}{2}\bar{\Sigma}_t H^\top H) {\xi}_t{\xi}^\top_t \ud t + \sigma_B\ud \bar{B}_t {\xi}^\top_t \\&+ {\xi}_t{\xi}^\top_t(A^\top -\frac{1}{2}H^\top H\bar{\Sigma}_t)\ud t  + {\xi}_t (\sigma_B\ud \bar{B}_t)^\top + \Sigma_B \ud t
	\end{align*}
	which, following the same procedure as for the conditional
	mean, leads to the sde for the conditional covariance
	$\bar{\Sigma}_t=\Expect[\xi_t\xi^\top_t|\clZ_t]$.  It is given by
	%	upon taking the conditional expectation (same as the procedure for the conditional mean), 
	\begin{equation*}
	\ddt \Sigmabar_t = A \Sigmabar_t + \Sigmabar_tA^T + \Sigma_B - \Sigmabar_tH^TH\Sigmabar_t
	\label{eq:Xii-variance}
	\end{equation*}
	identical to the Ricatti equation~\eqref{eq:kalman-variance}.
	% starting with the same initial condition. 
	Hence, because $\Sigmabar_0 = \Sigma_0$, $\Sigmabar_t =\Sigma_t$ for all $t\geq0$. 
	This also implies $\Kbar_t=\k_t$, which in turn implies that the sde for conditional mean \eqref{eq:Xii-mean} is identical to the Kalman filter equation~\eqref{eq:kalman-mean}. 
	Therefore, again because $\bar{m}_0 = m_0$, $\bar{m}_t = m_t$ for all $t\geq0$. 
	
	Given $\bar{\Sigma}_t= \Sigma_t$ and $\bar{m}_t = m_t$, the mean-field terms in the McKean-Vlasov sde~\eqref{eq:FPF-lin} can be treated as exogenous processes. Therefore, the McKean-Vlasov sde~\eqref{eq:FPF-lin} simplifies to a Ornstein-Uhlenbeck sde. Because the distribution of the initial condition $\bar{X}_0$ is Gaussian, the distribution of $\bar{X}_t$ is also Gaussian given by $\mathcal{N}(m_t,\Sigma_t)$ which is equal to the posterior distribution given by Kalman filter and concludes the proof. 

\end{proof}

\section{Proof of Prop.~\ref{prop:opt-sde-vector}}
\label{proof:opt-sde}
The key step in the proof is the following Lemma:
\begin{lemma}
	Consider the ode \eqref{eq:kalman-variance}. Let $\Sigma_t$ be
	its solution for $t\in [0,T]$.  
	Then 
	\begin{equation}
	\Sigma_{t + \Delta t}^{\half}
	(\Sigma_{t+ \Delta t}^{\half}
	\Sigma_t\Sigma_{t+\Delta t}^\half)^{-\half}
	\Sigma_{t+\Delta t}^\half =
	I + G_t\Delta t + O(\Delta t^2)
	\label{eq:approx-Sigma-vector}
	\end{equation}
	where $G_t$ is the solution to the matrix equation
	\begin{equation}
	G_t\Sigma_t + \Sigma_tG_t = A\Sigma_t+ \Sigma_t A^T + \Sigma_B - \Sigma_tH^TH\Sigma_t
	\label{eq:G-lemma}
	\end{equation}
	and the $O(\Delta t^2)$ in~\eqref{eq:approx-Sigma-vector} is
	uniformly bounded for all $t \in [0,T]$.
	\label{lemma:approx-Sigma-vector}
\end{lemma}
\begin{proof}
	From the theory of dynamic Riccati equations, the solution is 
	bounded over any finite time
	horizon~\cite{bishop2017stability}. Moreover, because
        $\Sigma_0\succ 0$, $\Sigma_t \succ 0$.   
	Fix $t \in [0,T]$, and define 
	\begin{equation*}
	F(s):=\Sigma_{t + s}^{\half}
	(\Sigma_{t+ s}^{\half}
	\Sigma_t\Sigma_{t+s}^\half)^{-\half}
	\Sigma_{t+s}^\half
	\end{equation*}
	Equation~\eqref{eq:approx-Sigma-vector} is obtained by
	considering the Taylor series of $F(s)$ at $s=0$
	\begin{equation*}
	F(\Delta t) = I + \dot{F}(0)\Delta t + \frac{1}{2}\ddot{F}(\tau)\Delta t^2
	\end{equation*} 
	and showing that
	$\dot{F}(0)=G_t$; here $\tau \in [0,\Delta t]$.  The second order term is uniformly bounded
	for all $t \in [0,T]$ because $\Sigma_t$ is positive definite
	and bounded.   In order to verify $\dot{F}(0)=G_t$, express
	\begin{equation*}
	F(s)\Sigma_tF(s)= \Sigma_{t+s}
	\end{equation*}  
	Evaluating the derivative with respect to $s$ at $s=0$
	\begin{equation*}
	\dot{F}(0)\Sigma_t + \Sigma_t\dot{F}(0) = A\Sigma_t+ \Sigma_t A^T + \Sigma_B - \Sigma_tH^\top H\Sigma_t
	\end{equation*}
	By the uniqueness property of the solution to the Lyapunov
	equation \eqref{eq:G-lemma}, $\dot{F}(0)=G_t$.
\end{proof}

\medskip

\begin{proof}(Prop.~\ref{prop:opt-sde-vector}) 
	The proof of exactness is similar to the proof of Thm.~\ref{thm:exactness} and is omitted. 
	In order to obtain the optimal transport sde, the time stepping procedure is used.
	The key step in the procedure is to obtain the optimal transport map $T_k$.
	The optimal map is between two Gaussians, $\NN(m_{t_\ii},\Sigma_{t_\ii})$ and $\NN(m_{t_{\ii+1}},\Sigma_{t_{\ii+1}})$.
	By Thm.~\ref{thm:opt-map-Gaussians}, the optimal map is,
	\begin{equation*}
	\Xbar_{t_{\ii+1}}=m_{t_{\ii+1}} + F_k(\Xbar_{t_\ii}-m_{t_{\ii}})
	\label{eq:XiiTempRd}
	\end{equation*} 
	where $F_\ii=\Sigma_{t_{\ii+1}}^{\half}(\Sigma_{t_{\ii+1}}^{\half}\Sigma_{t_{\ii}}\Sigma_{t_{\ii+1}}^\half)^{-\half}\Sigma_{t_{\ii+1}}^\half$.
	Using Lemma \ref{lemma:approx-Sigma-vector}, 
	\begin{equation*}
	\begin{aligned}
	\Xbar_{t_{\ii+1}}
	&=m_{t_{\ii+1}} + (\Xbar_{t_\ii} - m_{t_\ii}) +  G_{\ii}(\Xbar_{t_\ii} - m_{t_\ii})\Delta t + O(\Delta t^2)
	\end{aligned}
	\end{equation*}
	To obtain the sde, take a sum over $k=0,1,\ldots,n-1$, %Then as $\Delta t \to 0$ we would obtain
	\begin{equation*}
	\begin{aligned}
	\Xbar_{t_n}&=\Xbar_{t_0} + m_{t_n}-m_{t_0} + \sum_{\ii=0}^{n-1}\big[G_\ii(\Xbar_{t_\ii} - m_{t_\ii})\Delta t + O(\Delta t^2)\big]
	\end{aligned}
	\end{equation*}
	In the limit as $\Delta t \to 0$,
	\begin{equation*}
	\begin{aligned}
	\Xbar_{t_n}&= \Xbar_{t_0} + m_{t_n}-m_{t_0}+\int_{0}^{t}G_s(\Xbar_s-m_s)\ud s
	\end{aligned}
	\end{equation*}
	where the uniform boundedness of the second order term is used. 
	The associated sde is,
	\begin{equation*}
	\ud \Xbar_t = \ud m_t + G_t(\Xbar_t - m_t)\ud t
	\end{equation*}
	where $\ud m_t$ is given by
	\eqref{eq:kalman-mean}. 
	Finally one obtains \eqref{eq:opt-sde} by replacing
	$m_t$ and $\Sigma_t$ with $\mbar_t$ and $\Sigmabar_t$
	respectively.
	%, which are identical by exactness. 
\end{proof}

\section{Derivation of optimal FPF in singular covariance case
}\label{apdx:singular-case}

Consider the following general form of the controlled process:
\begin{equation}\label{eq:xbar_appdx}
\ud \Xbar_t =  G_t(\Xbar_t-\mbar_t) \ud t + \ud v_t+ \sigma_t \ud \bar{B}_t
\end{equation}  
The problem is to choose $G_t$, $v_t$  and $\sigma_t$ such  that the
stochastic map $\Xbar_t \to \Xbar_{t+\Delta t}$ is optimal in the
limit as $\Delta t \to  0$.  The optimality criterion
is the Kantorovich form~\eqref{eq:Kantorovich} of the optimal
transportation problem.  The particular choice~\eqref{eq:xbar_appdx} of
the sde for $\{\Xbar_t\}_{t\ge 0}$ is motivated by the optimal
transport sde~\eqref{eq:opt-sde} derived in
Prop.~\ref{prop:opt-sde-vector}.  We expect to recover the
deterministic form of the filter ($\sigma_t=0$) for the special
case when the covariance is non-singular.

The stochastic map $\Xbar_t \to \Xbar_{t+\Delta t}$ is given by
\begin{align*}
\Xbar_{t+\Delta t} &= \Xbar_t + \int_t^{t+\Delta t} G_s(\Xbar_s-\mbar_s)\ud s +(v_{t+\Delta t}-v_t) + \int_t^{t+\Delta t}\sigma_s\ud B_s\\&
=\Xbar_t + \Delta t G_t (\Xbar_t - \mbar _t) + v_{t+\Delta t} - v_t + \sqrt{\Delta t} \sigma_t \zeta  + o(\Delta t)
\end{align*}
where $\zeta \sim\calN(0,1)$. The stochastic map is optimal if (i) the
marginals $\Xbar_t \sim\calN(m_t,\Sigma_t)$ and $\Xbar_{t+\Delta t}
\sim\calN(m_{t+\Delta t}, \Sigma_{t+\Delta t})$, and (ii) the transport
cost $\Expect[|\Xbar_{t+\Delta t}-\Xbar_t|^2]$ is minimized. 

Now, given $\Xbar_t \sim\calN(m_t,\Sigma_t)$, the marginal constraint
is satisfied by the following choice:
\begin{align*}
&m_t + v_{t+\Delta t} - v_t = m_{t+\Delta t} \\
&(I + \Delta t G_t)\Sigma_t(I+\Delta t G_t) + \Delta t \sigma_t \sigma_t^\top + o(\Delta t) = \Sigma_{t+\Delta t}
\end{align*}
The first constraint simply means that the increment of $\nu_t$ must be chosen according to
\[
\ud v_t = \ud m_t = Am_t \ud t+ \K_t(\ud Z_t - H m_t\ud t)
\] 
Dividing the second constraint by $\Delta t$ and taking the limit as
$\Delta t \to 0$ gives 
\begin{equation}\label{eq:OT-FPF=G-sigma-constaint}
G_t \Sigma_t + \Sigma_t  G_t^\top + \sigma_t \sigma_t^\top = \Ricc(\Sigma_t)
\end{equation}
which means that, in the limit as $\Delta t\rightarrow 0$, $G_t$ and $\sigma_t$ must satisfy the
constraint~\eqref{eq:OT-FPF=G-sigma-constaint}.  Clearly, there are
infinitely many possible choices for $G_t$ and $\sigma_t$ which
accounts for the non-uniqueness of the control law as 
discussed in \Sec{sec:non-uniqueness}.  

A unique choice is obtained by minimizing the optimal transportation cost
\begin{align*}
\Expect[|\Xbar_{t+\Delta t}-\Xbar_t|^2] &= |m_{t+\Delta t} - m_t|^2 +\Delta t \underbrace{\trace(\sigma_t\sigma_t^\top)}_{f_1(\sigma_t)}\\&+ (\Delta t)^2 \underbrace{\trace(G_t\Sigma_tG_t^\top)}_{f_2(G_t)} + o(\Delta t^2)
\end{align*}
Taking the limit as $\Delta t \to 0$ suggests the following sequence
of problems: (i) Choose $\sigma_t$ to first minimize $f_1(\sigma_t)$;
and (ii) Choose $G_t$ to next minimize $f_2(G_t)$.  
These formal considerations lead to the following optimization problem:

%This justifies the optimization problem~\eqref{eq:singular-opt-objective}.
%
%
%In Appendix~\ref{apdx:singular-case}, we show that $\ud v_t = A \mbar_t\ud t + \Kbar_t(\ud Z_t -H\mbar_t) \ud t$, and $\sigma_t$ and $G_t$ are the solutions to the following optimization problem:

%The result, is the following Proposition.  The proof appears in
%. In order to state the result, it is necessary to define the following sets, and opttimization

\medskip

\newP{Optimization problem} Define
$f_1(\sigma):=\trace(\sigma\sigma^\top)$, $f_2(G):=\trace(G \Sigma
G^\top)$ together with the sets 
\begin{align*}
&\mathcal{D}_\Sigma \eqdef\{(\sigma,G) \in \Re^{d\times d_B}\times \Re^{d\times d} ;~G\Sigma + \Sigma G^\top  + \sigma \sigma^\top = \Ricc(\Sigma)\}\\
&\mathcal{D}_\Sigma\vert_{\sigma^*} \eqdef\{G \in \Re^{d\times d};~G\Sigma + \Sigma G^\top  + \sigma^*{\sigma^*}^\top = \Ricc(\Sigma)\}
\end{align*} 
The pair $(\sigma^*,G^*) \in \mathcal{D}_\Sigma$ are said to be {\em optimal} if 
\begin{equation}
\begin{aligned}
\trace(\sigma^*(\sigma^*)^\top) &= \min_{(\sigma,G)\in
  \mathcal{D}_\Sigma} f_1(\sigma),\\ \trace(G^* \Sigma G^*)&=
\min_{G\in \mathcal{D}_\Sigma\vert_{\sigma^\star}} f_2(G)
\end{aligned}
\label{eq:singular-opt-objective}
\end{equation}

Let $P_R$ and $P_K$ be  the orthogonal projection matrices onto the range and kernel space of $\Sigma$. 

\begin{proposition}\label{prop:OT-FPF-opt-sde-singular}
	Consider the optimization
        problem~\eqref{eq:singular-opt-objective}.  Its optimal
        solution  $(\sigma^*,G^*)$ is as follows: $\sigma^*=P_K \sigma_B$ and $G^*$ is the (unique such) symmetric solution of the matrix equation
\begin{align}
	%	\sigma^*_t &= e_t\\
	G^* \Sigma + \Sigma G^*& = \Ricc(\Sigma) - \sigma^* (\sigma^*)^\top  \label{eq:Lyapunov-G-e-appdx}
	%	A - \frac{1}{2}\bar{\Sigma}_t H H^\top  + \frac{1}{2}(\sigma_B + e_t) u_t^\top + \Omega^{(1)}_t
	\end{align}
	solve the optimization problem~\eqref{eq:singular-opt-objective}. 
\end{proposition}

\medskip

These choices yield the formula~\eqref{eq:opt-sde-singular} for the
optimal FPF described in \Sec{sec:singular-case}.  It remains to prove
the Proposition.  
	% \begin{equation}
	% \label{eq:opt-sde-singular}
	% \ud \Xbar_t = A\mbar_t + \Kbar_t(\ud Z_t - H\mbar_t   \ud t) + G^*_t(\Xbar_t-\mbar_t)\ud t + \sigma^*_t\ud \Bbar_t
	% \end{equation}
	%	 simultaneously, i.e  $f_1(\sigma^*_t) = \min_\sigma
        %	 f_1(\sigma)$ and $ f_2(G^*_t) = \min_G f_2(G)$. 

\medskip

\begin{proof} (of Prop.~\ref{prop:OT-FPF-opt-sde-singular}) 
% First, it is clear that $\sigma^*$ and $G^*$ satisfy the
% constraint~\eqref{eq:OT-FPF=G-sigma-constaint}. As a result
% $(\sigma^*,G^*) \in \mathcal{D}_\Sigma$. 
% Next, we show that $\sigma^* = P_K\sigma_B$ is  the minimizer for
% $\trace(\sigma\sigma^\top)$. 
For any $(\sigma,G) \in \mathcal{D}_\Sigma$, multiply both sides of
constraint~\eqref{eq:OT-FPF=G-sigma-constaint} from left and right by
$P_K$ to obtain
\begin{equation*}
P_K \sigma \sigma^\top P_K = P_K \sigma_B \sigma_B P_K 
\end{equation*}
where $P_K\Sigma=\Sigma P_K = 0$ is used. 
%Using this identity, the value of objective function is expressed as
Therefore,
\begin{align*}
f_1(\sigma)=\trace(\sigma \sigma^\top) &= \trace(P_K \sigma \sigma^\top P_K) + \trace(P_R \sigma \sigma^\top P_R)  \\&=
\trace(P_K \sigma_B \sigma_B^\top P_K) + \trace(P_R \sigma \sigma^\top P_R) \\&=
\trace(\sigma^* (\sigma^*)^\top) + \trace(P_R \sigma \sigma^\top P_R)
\end{align*}
The second term is non-negative and zero iff $\sigma=\sigma^*$.
Therefore, $\sigma^*$ minimizes $f_1(\sigma)$.  

It remains to show that $G^*$ minimizes $f_2(G)$ over all $G \in
 \mathcal{D}_\Sigma\vert_{\sigma^*}$.  We begin by showing that any
 symmetric solution of the Lyapunov
 equation~\eqref{eq:Lyapunov-G-e-appdx} exists and is well-defined.
 The formula for the solution is given by
\begin{equation*}
G^* = \int_0^\infty  e^{-t\Sigma} ( \Ricc(\Sigma) - P_K\sigma_B \sigma_B P_K )e^{-t\Sigma} \ud t + P_K \Omega^{(0)}P_K
\end{equation*}
where $\Omega^{(0)}$ is any symmetric matrix. 
The integral is well-defined because
\begin{align*}
G^*&-P_K \Omega^{(0)}P_K =(P_K + P_R)(G^*-P_K \Omega^{(0)}P_K)(P_K + P_R) \\&=   \int_0^\infty (P_K + P_R)
                                    e^{-t\Sigma} ( \Ricc(\Sigma) -
                                    P_K\sigma_B \sigma_B P_K
                                    )e^{-t\Sigma} (P_K + P_R) \ud t \\
%&= \int_0^\infty e^{-t\Sigma} P_K ( \Ricc(\Sigma) - P_K\sigma_B \sigma_B P_K ) P_K e^{-t\Sigma}   \ud t \\
% &+ \int_0^\infty e^{-t\Sigma} P_R ( \Ricc(\Sigma) - P_K\sigma_B \sigma_B P_K )e^{-t\Sigma}   \ud t 
% \\
% &+ \int_0^\infty e^{-t\Sigma} P_K ( \Ricc(\Sigma) - P_K\sigma_B \sigma_B P_K ) P_Re^{-t\Sigma}   \ud t \\
&=\int_0^\infty e^{-t\Sigma} P_R ( \Ricc(\Sigma) - P_K\sigma_B \sigma_B P_K )e^{-t\Sigma}   \ud t  \\
&\quad + \int_0^\infty e^{-t\Sigma} P_K ( \Ricc(\Sigma) - P_K\sigma_B \sigma_B P_K ) P_Re^{-t\Sigma}   \ud t 
\end{align*}
where  $P_K ( \Ricc(\Sigma) - P_K\sigma_B \sigma_B P_K ) P_K=0$ is
used. The integral is bounded because $\|e^{-t\Sigma} P_R\|_2 =
\|P_Re^{-t\Sigma} \|_2 = e^{-t\mu}$ where $\mu>0$ is the smallest
non-zero eigenvalue of $\Sigma$. 

It remains to show that $G^*$ thus defined is optimal.  Express an arbitrary 
$G\in\mathcal{D}_\Sigma\vert_{\sigma^*}$ as $G=G^* + V$.  Since $G, G^*
\in \mathcal{D}_\Sigma\vert_{\sigma^*}$, it is an easy calculation to
see that 
\begin{equation*}
V\Sigma + \Sigma V^\top = 0
\end{equation*}
Now, 
\begin{align*}
f_2(G) &= \trace(G\Sigma G^\top) = \trace((G^* + V) \Sigma (G^* +V)^\top) \\&= 
 \trace( G^* \Sigma G^* ) +  \trace( G^* \Sigma V^\top  ) +  \trace( V \Sigma G^*  ) +  \trace( V \Sigma V^\top  ) \\
 &=  \trace( G^* \Sigma G^* )  +  \trace( G^* (V \Sigma + \Sigma V^\top))  +  \trace( V \Sigma V^\top  ) \\
 &=  \trace( G^* \Sigma G^* ) +  \trace( V \Sigma V^\top  ) 
\end{align*}
%The second term is strictly positive. 
Therefore, the choice $G=G^*$ minimizes $f_2(G)$. 
 
\end{proof}

\newP{Justification for formula~\eqref{eq:singular-G-Omega}} The goal is  to show that any symmetric solution to the matrix equation~\eqref{eq:Lyapunov-G-e} is of the form~\eqref{eq:singular-G-Omega}. Without loss of generality, express  the solution as
\begin{equation*}
G_t = A - \frac{1}{2}\Sigmabar_t H H^\top + \frac{1}{2}P_R \Sigma_B \Sigmabar_t^+ + P_K\Sigma_B  \Sigmabar_t^+ + \Sigmabar_t^+ \Sigma_B P_K + \tilde{G}_t
\end{equation*} 
where  $\tilde{G}_t$ is the new variable. Because $G_t$ is symmetric, $\tilde{G}_t$ should satisfy 
\begin{equation}
\begin{aligned}
\tilde{G}_t -\tilde{G}_t^\top = &A^\top - A  +  \frac{1}{2}(\Sigmabar_t H H^\top -H^\top H \Sigmabar_t) \\&- \frac{1}{2}(P_R \Sigma_B \Sigmabar_t^+ - \Sigmabar_t^+ \Sigma_B P_R)
\end{aligned}
\label{eq:tilde-G-constraint}
\end{equation} 
Inserting this form of the solution for $G_t$ to the matrix equation~\eqref{eq:Lyapunov-G-e} yields:
\begin{align*}
&A  \Sigmabar_t  + \Sigmabar_t A^\top - \Sigmabar_t H H^\top \Sigmabar_t + \frac{1}{2}P_R \Sigma_B\Sigmabar_t^+\Sigmabar_t + \frac{1}{2}\Sigmabar_t  \Sigmabar_t^+  \Sigma_BP_R \\
&+P_K \Sigma_B \Sigmabar_t^+ \Sigmabar_t +    \Sigmabar_t \Sigmabar_t^+ \Sigma_B P_K + \tilde{G}_t \Sigmabar_t + \Sigmabar_t \tilde{G}_t^\top = \Ricc(\Sigmabar_t) - \sigma_t\sigma_t^\top
\end{align*} 
Using $\Sigmabar_t^+ \Sigmabar_t = \Sigmabar_t \Sigmabar_t^+ = P_R$, $\sigma_t = P_K \sigma_B$ and $P_K + P_R =I$ concludes:
\begin{align*}
 \tilde{G}_t \Sigmabar_t + \Sigmabar_t \tilde{G}_t^\top =0
\end{align*} 
All the solutions to this linear matrix equation can be expressed as:
\begin{equation*}
\tilde{G}_t = P_K\Omega^{(0)}P_K + P_R\Omega^{(1)}\Sigmabar_t^+ 
\end{equation*}
where $\Omega^{(0)} \in \Re^{d\times d}$ is arbitrary and $\Omega^{(1)} \in \Re^{d\times d}$ is skew-symmetric. Next, conditions for $\Omega^{(0)}$ and $\Omega^{(1)}$ are obtained so that the constraint~\eqref{eq:tilde-G-constraint} is true.  For $\Omega^{(0)}$, multiply~\eqref{eq:tilde-G-constraint} by $P_K$ from left and right to obtain 
\begin{equation*}
P_K \Omega^{(0)} P_K - P_K\Omega^{(0)}  P_K  = P_K (A^\top - A )P_K 
\end{equation*} 
This condition is satisfied when $\Omega^{(0)} + A$ is symmetric.   For $\Omega^{(1)}$, multiply~\eqref{eq:tilde-G-constraint} by $P_R$ from left and right to obtain 
%The solution to this equation can be expressed as $\tilde{G_t =  P_K \Omega^{(0)} P_R\Omega^{(1)} \Sigmabar_t^+$
%Using the decomposition $\tilde{G}_t = P_K \tilde{G}_t P_K +  P_K \tilde{G}_t P_R + P_R \tilde{G}_t P_K + P_R \tilde{G}_t P_R$, it can be shown that
%\begin{align*}
%&P_K \tilde{G}_t P_R  = P_R \tilde{G}_t P_K = 0 \\
%&P_R \tilde{G}_t P_R \Sigmabar_t + \Sigmabar_t P_R\tilde{G}_t^\top  P_R = 0
%\end{align*}
%Therefore $P_R \tilde{G}_t P_R = P_R\Omega^{(1)} \Sigmabar_t^+$ where $\Omega^{(1)}$ is a skew symmetric matrix. 
%Using the constraint~\eqref{eq:tilde-G-constraint}, it is concluded that $P_K\Omega^{(0)}P_K$
\begin{equation}
\begin{aligned}
%P_K \tilde{G}_t P_K - P_K\tilde{G}_t^\top P_K  = &P_K (A^\top - A )P_K \\
&P_R \Omega^{(1)} \Sigmabar_t^+  + \Sigmabar_t^+ \Omega^{(1)}_t P_R= P_R(A^\top - A)P_R \\&+  \frac{1}{2}P_R(\Sigmabar_t H H^\top  - H^\top H \Sigmabar_t)P_R \- \frac{1}{2}P_R(P_R \Sigma_B \Sigmabar_t^+ - \Sigmabar_t^+ \Sigma_B, P_R)P_R\label{eq:singular-Omega}
\end{aligned}
\end{equation}
This is the condition that is satisfied by $\Omega^{(1)}$.

\section{Proof of the
	Prop.~\ref{prop:conv_error}}\label{apdx:mean-var-error}

%\subsection{Background on the Stability of the Kalman filter}\label{apdx:KF-stability}

%\medskip

%The following result on filter stability can be found in~\cite{}:

The proof of the Prop.~\ref{prop:conv_error} relies on the stability
theory of the Kalman filter.  The following results are quoted without proof:

\begin{lemma}
	Consider the Kalman filter~\eqref{eq:kalman-mean}-\eqref{eq:kalman-variance}
	with initial condition $(m_0,\Sigma_0)$.  Then, under Assumption (I):
	\begin{romannum}
		\item (\cite[Thm. 4.11]{kwakernaak1972linear}) There exists a solution $\Sigma_{\infty}\succ 0$ to the
		algebraic Riccati equation (ARE)
		\begin{align}
		A \Sigma_\infty +  \Sigma_\infty A^\top + \sigma_B\sigma_B^{\top} -
		\Sigma_\infty H^\top H\Sigma_\infty = 0\label{eq:are}
		\end{align}
		such that $A - \Sigma_\infty H^\top H$ is Hurwitz.  Let 
		\begin{equation}
		\begin{aligned}
		0&<\lambda_0 = \text{min} \{ -\text{Real} \; \lambda : \lambda \in \text{Spec}( A - \Sigma_\infty H^\top H)\}
		\label{eq:lambda0_defn}
		\end{aligned}
		\end{equation}
		\item(\cite[Thm. 1.1]{bishop2017stability}) If the initial covariance matrix $\Sigma_0 \succ 0$, then there exists matrices 
		$\Lambda_{\text{min}}, \Lambda_{\text{max}} \succ 0$ such that the solution $\Sigma_t$ satisfies
		\begin{equation*}
		\Lambda_{\text{min}} \preceq \Sigma_t \preceq \Lambda_{\text{max}}
		\end{equation*}
		\item(\cite[Lem. 2.2]{ocone1996}) The error covariance $\Sigma_t \rightarrow \Sigma_{\infty}$
		exponentially fast for {\em any} initial condition
		${\Sigma}_0$ (not necessarily the prior): for all $\lambda\in (0,\lambda_0)$, there exists a constant $c_\lambda$ such that 
		\[
		\lim_{t\rightarrow \infty}  \Fnorm{\Sigma_t -
			\Sigma_{\infty}} \le  c_\lambda e^{-2\lambda t} \rightarrow 0
		\]
		\item\cite[Thm. 2.3]{ocone1996} Starting from two initial conditions $(m_0,\Sigma_0)$ and $(\tilde{m_0},\tilde{\Sigma}_0)$, the means converge in
		the following senses:
		\begin{align*}
		\lim_{t\rightarrow \infty}  {\sf E} [\|m_t - \tilde{m}_t\|_2^2] &\; \le \; \text{(const.)} \; e^{-2\lambda t} \rightarrow 0
		%& \; \le\; \text{(const.)}\; e^{-2\lambda_0 t} \rightarrow 0
		\\
		\lim_{t\rightarrow \infty} \|m_t - \tilde{m}_t\|_2 e^{\lambda
			t} &\; = \; 0 \quad {\text{a.s.}}
		\end{align*}
		for all $\lambda\in (0,\lambda_0)$. 
	\end{romannum}
	\label{lem:KF-stability}
\end{lemma}

\medskip

In the remainder of this paper, the notation $\Sigma_\infty$ is used to denote
the positive definite solution of the ARE~\eqref{eq:are} and $\lambda_0$ is used to
denote the spectral bound as defined in~\eqref{eq:lambda0_defn}. 

\medskip

\newP{Proof of Prop.~\ref{prop:conv_error}} 
%First, we derive the evolution equations for the empirical mean and empirical variance. 
Consider the finite-$N$ filter~\eqref{eq:Xit-d} for the deterministic FPF.  The empirical mean and covariance are defined in
Eq.~\eqref{eq:empr_app_mean_var}.  The error is defined as
\[
\xi^i_t := X^i_t - m^{(N)}_t \quad \text{for}\;\;i=1,2,\ldots,N
\]   
The evolution equations for the mean, covariance, and the error are as follows:  
\begin{subequations}
	\begin{align}
	\ud \mN_t &= A\mN_t \ud t + \k^{(N)}_t(\ud Z_t - H\mN_t\ud t)
	\label{eq:mean-evolution-d}\\
	%\ud \xi^i_t &= G_t^{(N)}\xi^i_t \ud t \\
	\frac{\ud \SigN_t}{\ud t} &= \Ricc(\SigN_t)\label{eq:var-evolution-d}\\
	\frac{\ud \xi^i_t}{\ud t} &= \sRicc(\SigN_t)\xi^i_t \label{eq:e-evolution-d}
	\end{align}
\end{subequations}

Eq.~\eqref{eq:mean-evolution-d} is
obtained by summing up Eq.~\eqref{eq:Xit-d} for the
$i^{\text{th}}$ particle from $i=1,\ldots,N$.
Equation~\eqref{eq:var-evolution-d} is obtained by application of It\"o rule
\begin{align*}
\ud (\xi^t{\xi^i_t}^\top) = \sRicc(\SigN_t)\xi^i_t{\xi^i_t}^\top\ud t + \xi^i_t{\xi^i_t}^\top \sRicc(\SigN_t) \ud t
%&(A + \frac{1}{2}\SigN_t{^{-1}} - \frac{1}{2}\mathsf{K}^{(N)}_tC)X^i_t {X^i_t}^\top \ud  t\\& +X^i_t{X^i_t}^\top(A + \frac{1}{2}\SigN_t{^{-1}} - \frac{1}{2}\mathsf{K}^{(N)}_tC)^\top\ud t
\end{align*}
and summing over $i=1,\ldots,N$ and dividing by $(N-1)$.  Equation~\eqref{eq:e-evolution-d} is obtained  by subtracting~\eqref{eq:Xit-d} for $X^i_t$ from~\eqref{eq:mean-evolution-d} for $\mN_t$.  
%It is noted that $G_t^{(N)}$ is well-defined because $\SigN_0$ and
%thus $\SigN_t$ is invertible because of Assumption (II).  

Because the equations for the empirical mean~\eqref{eq:mean-evolution-d}
and the empirical covariance~\eqref{eq:var-evolution-d} are identical to
the Kalman filter~\eqref{eq:kalman-mean}-\eqref{eq:kalman-variance}, the
a.s. convergence of mean and variance follows from
Lemma~\ref{lem:KF-stability} on the filter
stability theory.  It remains to derive the mean-squared estimates.
This is done in the following steps: 

\begin{enumerate}
	\item Denote $F_\infty := A-\Sigma_\infty H^\top H$.  In the
	step~1, an estimate for the spectral norm of the transition
	matrix $e^{tF_\infty}$ is obtained.  From
	Lemma~\ref{lem:KF-stability}, the eigenvalues of $F_\infty$
	have negative real parts smaller than $-\lambda_0$.
	% If $F_\infty$ was diagonalizable, we could have simply concluded the bound $\|\Phi_{t,s}\|\leq e^{-\lambda_0 t}$.  However, $F_\infty$ is not diagonalizable in general. 
	Consider the Jordan decomposition $J = P^{-1}F_\infty P$ to bound
	\begin{equation*}
	\|e^{tF_\infty}\|_2 \leq \|P\|_2\|P^{-1}\|_2 \left(\max_{0\leq k\leq n}~\frac{t^k}{k!}\right)e^{-\lambda_0t},\quad \forall t>0
	\end{equation*}
	where $n$ the largest multiplicity of the eigenvalues of $F_\infty$. As a result, for all $\lambda < \lambda_0$, there exists a constant $c'_\lambda := \|P\|_2\|P^{-1}\|_2  \sup_{t\geq 0}e^{-(\lambda_0-\lambda)t} \left(\max_{0\leq k\leq n}~\frac{t^k}{k!}\right)$ such that %As a result  
	\begin{equation*}
	\|e^{tF_\infty}\|_2 \leq c'_\lambda e^{-\lambda t}
	\end{equation*}	
	
	\item  % Consider the linear system $\frac{\ud}{\ud t}\Phi_{t,s}
	% = F_t \Phi_{t,s},~\Phi_{s,s}=I$ where $F_t= A-\Sigma_tH^\top
	% H$.  For $x_t = \Phi_{t,s}x_s$ we have 
	Denote $F_t:= A-\Sigma_tH^\top
	H$ and consider the linear system
	\begin{equation}\label{eq:xtFtxt}
	\frac{\ud}{\ud t}x_t = F_tx_t = F_\infty x_t +(\Sigma_\infty-\Sigma_t)H^\top H x_t
	\end{equation}
	Therefore
	\begin{equation*}
	x_t = e^{tF_\infty}x_s + \int_s^t e^{(t-\tau)F_\infty}(\Sigma_\infty-\Sigma_\tau)H^\top H x_\tau \ud \tau
	\end{equation*}	
	Upon taking the norm and using the triangle inequality
	\begin{align*}
	\|x_t\|_2\leq &c_\lambda e^{-t\lambda}\|x_s\|_2 \\&+ \int_s^t c_\lambda e^{-(t-\tau )\lambda} \|\Sigma_\tau-\Sigma_\infty\|_2\|H^\top H\|_2\|x_\tau\|_2\ud \tau
	\end{align*}
	The Gronwall inequality is then used to conclude that
	\begin{align*}
	\|x_t\|_2\leq c'_\lambda e^{-\lambda (t-s)}  \|x_s\|_2 e^{c'_\lambda \|H^\top H\|_2 \int_s^t \|\Sigma_\tau -\Sigma_\infty\|_2 \ud \tau }
	\end{align*}
	This shows that the transition matrix $\Phi_{t,s}$ for the
	linear system~\eqref{eq:xtFtxt} is bounded as follows:
	\begin{equation*}
	\|	\Phi_{t,s}\|_2\leq c'_\lambda e^{-\lambda (t-s)} e^{c'_\lambda \|H^\top H\|_2 \int_s^t \|\Sigma_\tau -\Sigma_\infty\|_2 \ud \tau }
	\end{equation*}
	%	Similarly, one concludes \[\|\Phi_{t,s}\|\leq C_\lambda^3 e^{\frac{1}{2\lambda}C_\lambda^3 \|\Sigma_s-\Sigma_\infty\|\|HH^\top\|} e^{-\lambda (t-s)}\leq C_\lambda^3 e^{\frac{1}{2\lambda}C_\lambda^3 \|\Sigma_0-\Sigma_\infty\|\|HH^\top\|} e^{-\lambda (t-s)}\] 
	Now, because of the exponential convergence $\|\Sigma_t -
	\Sigma_\infty\|_2\leq c_\lambda e^{-2\lambda t}$ from
	Lemma~\ref{lem:KF-stability}, 
	\begin{equation*}
	\|	\Phi_{t,s}\|_2\leq c'_\lambda e^{-\lambda (t-s)} e^{c'_\lambda  \|H^\top H\|_2 \frac{c_\lambda \|\Sigma_0-\Sigma_\infty\|_2}{2\lambda} }
	\end{equation*}
	and therefore $\|\Phi_{t,s}\|_2\leq c_\lambda e^{-\lambda (t-s)}$ for
	some constant $c_\lambda$.   
	\item Consider the empirical counterpart of the linear
	system~\eqref{eq:xtFtxt} defined using $F_t^{(N)}:=
	A-\Sigma_t^{(N)} H^\top
	H$.  Denote the associated transition matrix as
	$\Phi^{(N)}_{t,s}$.  Then, because $\SigN_t$ also evolves
	according to the Riccati equation and converges
	exponentially to $\Sigma_\infty$, by repeating the steps
	above, we also obtain $\|\Phi^{(N)}_{t,s}\|_2\leq c_\lambda
	e^{-\lambda (t-s)}$.  
	% Similarly, the spectral bound $\|\Phi^{(N)}_{t,s}\|_2\leq c_\lambda e^{-\lambda (t-s)}$ holds when one replaces $\Sigma_t$ with $\SigN_t$,  because $\SigN_t$ also evolves according to the Riccati equation and converges exponentially to $\Sigma_\infty$. 
	%evolves according to the 
	\item We are now ready to derive an estimate for the error
	$\SigN_t-\Sigma_t$. From the Riccati equation,
	\begin{align*}
	\frac{\ud}{\ud t}(\SigN_t-\Sigma_t)= &(A-\Sigma_t HH^\top)(\SigN_t-\Sigma_t) \\&+ (\SigN_t-\Sigma_t)(A-\SigN_t HH^\top)^\top
	\end{align*}
	whose solution is given by
	\begin{equation*}
	\SigN_t-\Sigma_t = \Phi_{t,0}(\SigN_0-\Sigma_0)(\Phi_{t,0}^{(N)})^\top
	\end{equation*}
	% where $\Phi_{t,0}$ and $\Phi_{t,0}^{(N)}$ are the state transition matrices that were defined in step 2. 
	% Then, 
	Therefore,
	\begin{align*}
	\|\SigN_t-\Sigma_t\|_F &\leq \| \Phi^{\Sigma_t}_t\|_2\| \Phi^{\SigN_t}_t\|_2 \|\SigN_0-\Sigma_0\|_F\\&\leq c^2_\lambda  e^{-2\lambda t} \|\SigN_0-\Sigma_0\|_F
	\end{align*}
	%	where we applied the upper-bound on the spectral norm of the transition matrix from step 2.
	Upon squaring and taking the expectation of both sides 
	\begin{align*}
	\Expect[\|\SigN_t-\Sigma_t\|^2_F] &\leq c^4_\lambda  e^{-4\lambda t} \Expect[ \|\SigN_0-\Sigma_0\|^2_F]\\
	% &= c_\lambda^4 e^{-4\lambda t} \Expect[\trace((\SigN_0-\Sigma_0)^2)]\\
	% &= c_\lambda^4 e^{-4\lambda t} \Expect[\trace((\frac{1}{N}\sum_{i=1}^N\xi^i_0{\xi^i_0}^\top - \Sigma_0)^2)]\\
	&= c_\lambda^4 e^{-4\lambda t} \frac{1}{N}\Expect[\trace((\xi^i_0{\xi^i_0}^\top - \Sigma_0)^2)]\\
	&\leq c_\lambda^4 e^{-4\lambda t} \frac{\Expect[\|\xi^i_0\|_2^4]}{N}= c_\lambda^4 e^{-4\lambda t} \frac{3\tr(\Sigma_0)^2}{N}
	\end{align*}	
	%where we defined matrix valued random variables $\zeta_i = X^i_0{X^i_0}^\top - \Sigma_0$. 
	\item Finally, a bound for the mean-squared error in estimating the
	mean is derived. Subtracting~\eqref{eq:kalman-mean} for the
	conditional mean from~\eqref{eq:mean-evolution-d} for the empirical mean yields:
	\begin{align*}
	\ud \mN_t - \ud m_t = &(A - \SigN_tH^\top H) (\mN_t-m_t)\ud t \\&+ (\SigN_t-\Sigma_t)H^\top H \ud I_t
	\end{align*} 
	where $\ud I_t  =\ud Z_t - H m_t \ud t$ is the increment of the
	innovation process. Its solution is given by
	\begin{align*}
	\mN_t - m_t = &\Phi^{(N)}_t (\mN_0 - m_0) \\&+ \int_0^t \Phi^{(N)}_{t,s}(\SigN_s-\Sigma_s)H^\top H \ud I_s 
	\end{align*}
	The norm of the first term is bounded by:
	\begin{align*}
	\Expect[\| \Phi^{(N)}_t (\mN_0 - m_0)\|_2^2 ]&\leq c_\lambda^2e^{-2\lambda t} \Expect[\|\mN_0-m_0\|_2^2]\\&\leq c_\lambda^2e^{-2\lambda t}  \frac{\trace(\Sigma_0)}{N}
	%\Expect\left[\left\|\int_0^t \Phi^{(N)}_{t,s}(\SigN_s-\Sigma_s)H^\top \ud I_s\right\|_2^2\right]& = \int_0^t \Expect\left[\trace\left(\Phi_{t,s}^{(N)} (\SigN_s-\Sigma_s) H^\top H(\SigN_s-\Sigma_s) {\Phi_{t,s}^{(N)}}^\top  \right) \right]\ud s \\
	%&\leq \int_0^t \Expect[\|\Phi_{t,s}^{(N)} (\SigN_s-\Sigma_s)\|^2_2] \| H\|^2_F   \ud s\\
	%%&\leq \| H\|^2_F \int_0^t \|\Phi_{t,s}^{(N)}\|_2^2 \Expect[\| (\SigN_s-\Sigma_s)\|^2_2]\ud s\\
	%&\leq  \| H\|^2_F \int_0^t c_\lambda^2e^{-2\lambda(t-s)} c_\lambda^4e^{-4\lambda s} \Expect[\| \SigN_0-\Sigma_0\|^2_2]\ud s\\
	%&= c_\lambda^6  \| H\|^2_F \frac{3\|\Sigma_0 \|_F^2}{N} ~\frac{e^{-2\lambda t}}{2\lambda }
	\end{align*}
	The norm of the second term is bounded by:
	\begin{align*}
	%\Expect[\| \Phi^{(N)}_t (\mN_0 - m_0)\|_2^2 ]&\leq c_\lambda^2e^{-2\lambda t} \Expect[\|\mN_0-m_0\|_2^2]\\&\leq c_\lambda^2e^{-2\lambda t}  \frac{\trace(\Sigma_0)}{N}\\
	\Expect&\left[\left\|\int_0^t \Phi^{(N)}_{t,s}(\SigN_s-\Sigma_s)H^\top \ud I_s\right\|_2^2\right] = \\&\int_0^t \Expect\left[\trace\left(\Phi_{t,s}^{(N)} (\SigN_s-\Sigma_s) H^\top H(\SigN_s-\Sigma_s) {\Phi_{t,s}^{(N)}}^\top  \right) \right]\ud s \\
	&\leq \int_0^t \Expect[\|\Phi_{t,s}^{(N)} (\SigN_s-\Sigma_s)\|^2_F] \| H\|^2_2   \ud s\\
	%&\leq \| H\|^2_F \int_0^t \|\Phi_{t,s}^{(N)}\|_2^2 \Expect[\| (\SigN_s-\Sigma_s)\|^2_2]\ud s\\
	&\leq  \| H\|^2_2 \int_0^t c_\lambda^2e^{-2\lambda(t-s)} c_\lambda^4e^{-4\lambda s} \Expect[\| \SigN_0-\Sigma_0\|^2_F]\ud s\\
	&\leq c_\lambda^6  \| H\|^2_2 \frac{3\tr(\Sigma_0)^2}{N} ~\frac{e^{-2\lambda t}}{2\lambda }
	\end{align*}
	where we used the fact that the innovation process is a Brownian motion~\cite[Lemma 5.6]{xiong2008} and It\^o isometry in the first step. 
	Adding the two bounds,
	\begin{align*}
	\Expect[\|m_t-\mN_t\|_2^2] \leq &e^{-2\lambda t}  \frac{2c_\lambda^2\trace(\Sigma_0)}{N} \\&+  e^{-2\lambda t}\frac{6c_\lambda^6  \| H\|^2_2\tr(\Sigma_0)^2}{2\lambda N} 
	\end{align*}

\end{enumerate}

\section{Proofs of the Prop.~\ref{prop:prop-chaos} and Cor.~\ref{cor:prop-chaos}}\label{apdx:prop-chaos}
\begin{proof} 
	%	Assume $\exists\alpha_1,\alpha_2>0$ such that $\alpha_1 I \preceq \Sigma_t \preceq \alpha_2 I$  and $\alpha_1 I \preceq \SigN_t \preceq \alpha_2 I$ for all $t>0$. 
	%	Express the error as $X^i_t - \bar{X}^i_t =(\mN_t - m_t) +  (X^i_t - \bar{\Xbar}^i_t)$.  The mean-squared error of the first term is bounded according to~\eqref{eq:mean-estimate}. The objective is to bound the mean-squared error of the second term. 
	%	By definition
	In the proof $S$ is used to denote $S_{\Lambda_0}$. 
	Use the decomposition 
	\begin{equation*}
	X^i_t = \mN_t + \xi^i_t,\quad \bar{X}^i_t = \bar{m}_t + \xibar^i_t
	\end{equation*}
	to bound the error as
	\begin{align*}
	\Expect[\|X^i_t-\bar{X}^i_t\|_2^2 \mathds{1}_{S}]^{1/2} &\leq
          \Expect[\|\mN_t-\bar{m}_t\|_2^2]^{1/2} +
            \Expect[\|\xi^i_t-\xibar^i_t\|_2^2
            \mathds{1}_{S}]^{1/2}\\
&\leq
          \frac{\text{(const.)}}{\sqrt{N}} +
            \Expect[\|\xi^i_t-\xibar^i_t\|_2^2
            \mathds{1}_{S}]^{1/2}
	\end{align*}
where we have used the exactness property $\bar{m}_t=m_t$ and the
bound~\eqref{eq:mean-estimate} derived in
Prop.~\ref{prop:conv_error}-(ii).  The hard part of the proof is to establish the following bound:
%	\begin{subequations}
		\begin{align}
%		\Expect[\|\mN_t-\bar{m}_t\|_2^2  \mathds{1}_{S}]^{\half} &\leq  \frac{\text{(const.)}}{\sqrt{N}} \label{eq:proof-m-bound}\\ 
		\Expect[\|\xi^i_t-\xibar^i_t\|_2^2  \mathds{1}_{S}]^{\half} &\leq \frac{\text{(const.)}}{\sqrt{N}} \label{eq:proof-xi-bound}
		\end{align}	
This is done next.
%	\end{subequations}

	% Equation~\eqref{eq:proof-m-bound}: The bound follows from the exactness property $\bar{m}_t=m_t$, the bound  $	\Expect[\|\mN_t-\bar{m}_t\|_2^2  \mathds{1}_{S}] \leq \Expect[\|\mN_t-\bar{m}_t\|_2^2] $, and the
	% bound~\eqref{eq:mean-estimate}.
	
The two processes evolve as follows:
	\begin{align*}
	\ud \xi^i_t &= \sRicc(\SigN_t)\xi^i_t \ud t,\quad \xi^i_0 = X^i_0 - \mN_0\\	
	\ud \xibar^i_t &= \sRicc(\Sigmabar_t)\xibar^i_t\ud t,\quad \xibar^i_0 = X^i_0 - \mbar_0 
	\end{align*} 
To express the solution, define the state
	transition matrix according to
	\begin{equation*}
	\frac{\ud}{\ud t} \Psi^{(Q_t)}_{t,s} = \sRicc(Q_t) \Psi^{(Q_t)}_{t,s},\quad \Psi_{s,s}^{(Q_t)} = I
	\end{equation*}	
Using this definition,
	\begin{equation}
	\begin{aligned}
	\xi^i_t - \xibar^i_t = &\Psi^{(\SigN_t)}_{t,0} (\xi^i_0 - \xibar^i_0) \\&+ \int_0^t \Psi_{t,s}^{(\SigN_s)}(\sRicc(\Sigmabar_s)-\sRicc(\SigN_s))\xibar^i_s\ud s
	\end{aligned}
	\label{eq:xi-xibar-integral}
	\end{equation}
	%	where  $\sRicc(Q) := A + \frac{1}{2}\Sigma_BQ^{-1} - \frac{1}{2} Q H^\top H$.
	% where the notation $\Psi^{(Q_t)}_{t,s}$ is used to denote the state
	% transition matrix defined according to
	% \begin{equation*}
	% \frac{\ud}{\ud t} \Psi^{(Q_t)}_{t,s} = \sRicc(Q_t) \Psi^{(Q_t)}_{t,s},\quad \Psi_{s,s}^{(Q_t)} = I
	% \end{equation*}
	
We claim: There exists $c_1,c_2>0$ and a matrix $\Lambda_{\text{min}} \succ 0$ such that for all $t\geq s \geq 0$:
	\begin{subequations}
		\begin{align}
		&\P(\{\SigN_t \succ \Lambda_{\text{min}}\} \cap  S)=1 ~\label{eq:SigN-lower-bound}\\
                &\P( \{\|\Psi_{t,s}^{(\SigN_t)}\|_2\leq c_1\}  \cap  S)=1\label{eq:proof-Psi-bound}\\ 
		&\Expect[\|\sRicc(\Sigmabar_t)-\sRicc(\SigN_t)\|^2_2 \mathds{1}_{S}]^\half \leq  \frac{c_2}{\sqrt{N}}e^{-2\lambda t}  \label{eq:proof-sRicc-bound}
		\end{align}	
	\end{subequations}
Assuming the claim is true, the bound~\eqref{eq:proof-xi-bound} is
obtained as follows:
%	 that if the vent $S$ is true, then  
%	\begin{align*}
%	\|\xi^i_t-\xibar^i_t\|_2\leq  &c_1 \|\xi^i_0 - \xibar^i_0\|_2 + c_1 \int_0^t  \frac{c_2}{\sqrt{N}} e^{-2\lambda t}\|\xibar^i_s\|\ud s
%	\end{align*}
%	Upon taking the mean-squared norm of both sides and using the triangle inequality
	\begin{align*}
	\Expect[&\|\xi^i_t-\xibar^i_t\|^2_2\mathds{1}_S]^{\half}\leq c_1 \Expect[\| \xi^i_0 -\xibar^i_0\|_2^2\mathds{1}_S]^{\half} \\&\hspace*{80pt}+ c_1  \frac{c_2}{\sqrt{N}} \int_0^t e^{-2\lambda t}\Expect[\|\xibar^i_s\|^2]^{\half}\ud s\\
	&\leq c_1 \Expect[\| \mN_0 -\mbar_0\|_2^2]^{\half} + c_1  \frac{c_2}{\sqrt{N}} \int_0^t e^{-2\lambda t}\tr(\Sigmabar_s)^\half\ud s\\
	&\leq   c_1 \frac{\trace(\Sigma_0)^\half}{\sqrt{N}} + c_1  \frac{c_2}{\sqrt{N}} \frac{\tr(\Lambda_\text{max})^\half}{2\lambda}
	\end{align*}	 
	where we used the identity $\xi^i_0 - \xibar^i_0 = \mbar_0 -
        \mN_0$ and $\Sigmabar_t = \Sigma_t \prec \Lambda_{\text{max}}$
        from Lemma~\ref{lem:KF-stability}-(ii).  

%This concludes~\eqref{eq:proof-xi-bound}. It remains to show~\eqref{eq:proof-Psi-bound},~\eqref{eq:proof-sRicc-bound}, and~\eqref{eq:SigN-lower-bound}. 
	
It remains to prove the claim.  The three
bounds~\eqref{eq:proof-Psi-bound}-\eqref{eq:SigN-lower-bound} are
obtained in the following three steps:
\begin{enumerate}
\item (Bound~\eqref{eq:SigN-lower-bound}):  On the event $S$, $\SigN_0
  \succ \Lambda_0$.  Let $\Lambda_t$ be the solution of the Riccati
  equation initialized at $\Lambda_0$. By
  Lemma~\ref{lem:KF-stability}-(ii), $\exists \; \Lambda_\text{min}\succ
  0$ such that $\Lambda_t \succ \Lambda_{\text{min}}$. Because
  $\SigN_0 \succ \Lambda_0$, by the monotonicity property of the
  solution to Riccati equation (see~\cite[prop. 4.1]{bishop2017stability}), $\SigN_t
  \succ \Lambda_t$. Therefore, $\SigN_t\succ \Lambda_{\text{min}}$.

\item (Bound~\eqref{eq:proof-Psi-bound}): We begin by deriving a bound for $\|\Psi_{t,s}^{(\Sigma_\infty)}\|_2$. Consider the linear system: 
	\begin{equation*}
	\frac{\ud}{\ud t } x_t = \sRicc(\Sigma_\infty)^\top x_t
	\end{equation*} 
	and a function $V(x) := x^\top \Sigma_\infty x$. Then, 
	\begin{align*}
	\frac{\ud}{\ud t} V(x_t) = x_t^\top \sRicc(\Sigma_\infty)\Sigma_\infty x_t + x_t^\top \Sigma_\infty \sRicc(\Sigma_\infty)^\top x_t = 0
	\end{align*}
	where we used the definition~\eqref{eq:sRicc-def}. As a result, for all $t \geq s\geq 0$, 
	\begin{align*}
	x_t^\top\Sigma_\infty x_t = x_s^\top\Sigma_\infty x_s\quad &\Rightarrow \quad \|x_t\|_2^2 \leq  \frac{\lambda_{\text{max}}(\Sigma_\infty)}{\lambda_{\text{min}}(\Sigma_\infty)}\|x_s\|_2^2\\
	&\Rightarrow  \quad \|(\Psi^{(\Sigma_\infty)}_{t,s})^\top x_s\|_2 \leq c_3 \|x_s\|_2\\
	% &\Rightarrow  \|(\Psi^{(\Sigma_\infty)}_{t,s})^\top\|_2 \leq c \\
	&\Rightarrow  \quad \|(\Psi^{(\Sigma_\infty)}_{t,s})\|_2 \leq c_3
	\end{align*}
	where the constant $c_3:=\sqrt{\frac{\lambda_{\text{max}}(\Sigma_\infty)}{\lambda_{\text{min}}(\Sigma_\infty)}}$. 
	
To obtain a bound for 
	$\|\Psi_{t,s}^{(\SigN_t)}\|_2$, consider the linear system
	\begin{align*}
	\frac{\ud}{\ud t} x_t &= \sRicc(\SigN_t)x_t \\&= \sRicc(\Sigma_\infty)x_t + (\sRicc(\SigN_t) - \sRicc(\Sigma_\infty))x_t
	\end{align*}  
	whose solution is given by
	\begin{align*}
	x_t = \Psi_{t,0}^{(\Sigma_\infty)} x_0 + \int_0^t \Psi_{t,s}^{(\Sigma_\infty )} (\sRicc(\SigN_s) - \sRicc(\Sigma_\infty))x_s\ud s
	\end{align*}
Therefore, using the bound $\|\Psi_{t,s}^{(\Sigma_\infty)}\|_2 \leq c_3 $,
	\begin{align*}
	\|x_t\|_2 \leq  c_3 \|x_0\|_2 + c_3 \int_0^t  \|\sRicc(\SigN_s)-\sRicc(\Sigma_\infty)\|_2 \|x_s\|_2\ud s
	%\\
	% \leq  c \|x_0\|_2 + c c_3 c_\lambda  \int_0^t c_3 \|\SigN_s-\Sigma_\infty\|_2 \|x_s\|_2\ud s
	\end{align*}
	Now, 
	\begin{align*}
	\|&\sRicc(\SigN_t) - \sRicc(\Sigma_\infty) \|_2 \\&=
	\|\frac{1}{2}\Sigma_B((\SigN_t)^{-1}-\Sigma^{-1}_\infty) -
                                                            \frac{1}{2}(\SigN_t-\Sigma_\infty)H^\top
                                                            H\|_2\\&\leq
                                                                     \frac{1}{2}(\|\Sigma_B\|_2
                                                                     \|\Sigma_\infty^{-1}\|_2
                                                                     \|(\SigN_t)^{-1}\|_2
                                                                     +
                                                                     \|H^\top
                                                                     H\|_2)
                                                                     \|\SigN_t
                                                                     -
                                                                     \Sigma_\infty\|_2\\
&\leq \;\; \underbrace{(\|\Sigma_B\|_2 \|\Sigma_\infty^{-1}\|_2
  \|\Lambda_{\text{min}}^{-1} \|_2  + \|H^\top H\|_2)
  c_{\lambda}}_{c_4} \;\; e^{-2\lambda t}
	\end{align*}
where we used Lemma~\ref{lem:KF-stability}-(iii) (because $\SigN_t$ is
a solution of the Riccati
equation~\eqref{eq:var-evolution-d}) and $ \|(\SigN_t)^{-1}\|_2\leq
\|\Lambda_{\text{min}}^{-1} \|_2 $ from~\eqref{eq:SigN-lower-bound}.  
	% \begin{align*}
	% \|&\sRicc(\SigN_t) - \sRicc(\Sigma_\infty) \|_2 \leq c_3 e^{-2\lambda t}
	% \end{align*}
%	where we have used AMIR and $c_3:=(\|\Sigma_B\|_2 \|\Sigma_\infty^{-1}\|_2 \|\Lambda_{\text{min}}^{-1} \|_2  + \|H^\top H\|_2)c_\lambda$. 
	Therefore,
	\begin{align*}
	\|x_t\|_2 \leq  c_3 \|x_0\|_2 + c_3 c_4  \int_0^t  e^{-2\lambda s}  \|x_s\|_2\ud s
	%\\
	% \leq  c \|x_0\|_2 + c c_3 c_\lambda  \int_0^t c_3 \|\SigN_s-\Sigma_\infty\|_2 \|x_s\|_2\ud s
	\end{align*}
	By an application of the Gr\"onwall inequality
	\begin{align*}
	\|x_t\|_2&\leq c_3 e^{c_3 c_4  \int_0^t e^{-2\lambda s}\ud s}\|x_0\|_2 \leq \underbrace{c_3 e^{c_3 c_4\frac{1}{2\lambda} }}_{=:c_1} \|x_0\|_2 
	\end{align*}
which implies $\|\Psi^{(\SigN_t)}_{t,s}\|\leq c_1 $ for all $t \geq s
\geq 0$.  This is the bound~\eqref{eq:proof-Psi-bound}.

\item (Bound~\eqref{eq:proof-sRicc-bound}): We have
	\begin{align*}
	\|&\sRicc(\SigN_t) - \sRicc(\Sigmabar_t) \|_2 = \|\sRicc(\SigN_t) - \sRicc(\Sigma_t) \|_2 
	\\&\leq \frac{1}{2}(\|\Sigma_B\|_2 \|\Sigma_t^{-1}\|_2 \|{\SigN_t}^{-1}\|_2  + \|H^\top H\|_2) \|\SigN_t - \Sigma_t\|_2 \\
	&\leq  \frac{1}{2}(\|\Sigma_B\|_2 \|\Lambda_{\text{min}}^{-1} \|_2^2  + \|H^\top H\|_2)  \|\SigN_t - \Sigma_t\|_F
%	e^{-2\lambda t}\ \frac{\text{(const.)}}{\sqrt{N}}=:\frac{c_2}{\sqrt{N}}e^{-2\lambda t}
	\end{align*} 
	where we used $\Sigmabar_t=\Sigma_t$ (by the exactness property), $\|\SigN_t-\Sigma_t\|_2\leq \||\SigN_t-\Sigma_t\|_F$, and $ \|(\SigN_t)^{-1}\|_2\leq
	\|\Lambda_{\text{min}}^{-1} \|_2 $ from~\eqref{eq:SigN-lower-bound}.  
	 Taking the squared expectation and
         using~\eqref{eq:var-estimate} in
         Prop.~\ref{prop:conv_error}-(ii) gives~\eqref{eq:proof-sRicc-bound}.
	
\end{enumerate}
\end{proof}

\medskip

\begin{proof}[Proof of the Corollary~\ref{cor:prop-chaos}]
	Consider the event  $S = S_{\half \Sigma_0}$. 
% Then, decompose the error to  two terms:
	% \begin{align*}
	% 	\left|\frac{1}{N}\sum_{i=1}^N f(X^i_t) -\Expect[f(X_t) |\clZ_t]\right| &= 	\left|\frac{1}{N}\sum_{i=1}^N f(X^i_t) -\Expect[f(X_t) |\clZ_t]\right|\mathds{1}_{S}\\& + 	\left|\frac{1}{N}\sum_{i=1}^N f(X^i_t) -\Expect[f(X_t) |\clZ_t]\right| \mathds{1}_{S^c}
	% \end{align*}
We derive the following bounds:
	\begin{align}
		&\Expect\left[
	\left|\frac{1}{N}\sum_{i=1}^N f(X^i_t) -\Expect[f(X_t) |\clZ_t]\right|^2\mathds{1}_S\right]^{1/2} \leq \frac{\text{(const.)}}{\sqrt{N}}\label{eq:bound-prop-chaos-S}\\
		&\Expect\left[\left|\frac{1}{N}\sum_{i=1}^N f(X^i_t) -\Expect[f(X_t) |\clZ_t]\right|^2\mathds{1}_{S^c}\right]^{1/2} \leq \frac{\text{(const.)}}{\sqrt{N}}\label{eq:bound-prop-chaos-Sc}
	\end{align}
\begin{enumerate}
\item (Bound~\eqref{eq:bound-prop-chaos-S}) Using the triangle
  inequality,
	\begin{equation*}
	\begin{aligned}
	\Expect&\left[
	\left|\frac{1}{N}\sum_{i=1}^N f(X^i_t) -\Expect[f(X_t)
          |\clZ_t]\right|^2\mathds{1}_S\right]^{1/2}
      \\&\quad\leq\Expect\left[\left|\frac{1}{N}\sum_{i=1}^N f(X^i_t)
          -\frac{1}{N}\sum_{i=1}^Nf(\bar{X}^i_t)\right|^2\mathds{1}_S\right]^{1/2}
      \\& \quad \quad \quad + \;\; \Expect\left[\left|\frac{1}{N}\sum_{i=1}^N f(\bar{X}^i_t) -\Expect[f(X_t) |\clZ_t]\right|^2\right]^{1/2} 
	\end{aligned}
	\end{equation*}
	Now, because $\bar{X}^i_t$ are i.i.d with distribution equal
        to the conditional distribution, the second term on the
        right-hand side 
	\begin{equation*}
	\Expect\left[\left|\frac{1}{N}\sum_{i=1}^N f(\bar{X}^i_t) -\Expect[f(X_t) |\clZ_t]\right|^2\right]^{1/2} = \frac{
		%		\Expect\left[\left|f(X_t) - \Expect[f(X_t)|\clZ_t]
		%		\right|^2|\clZ_t
		%		\right]
		\text{Var}(f(X_t)|\clZ_t)
	}
	{\sqrt{N}}
	% \label{eq:varf}
	\end{equation*}
The first term on the right-hand side is bounded as follows:
	%By the decomposition~\eqref{eq:error-f-decompose} and the bound~\eqref{eq:varf} it only remains to bound
	\begin{align*}
	\Expect&\left[\left|\frac{1}{N}\sum_{i=1}^N f(X^i_t)
                 -\frac{1}{N}\sum_{i=1}^N
                 f(\bar{X}^i_t)\right|^2\mathds{1}_S\right]^{1/2}
          \\&\quad \leq \frac{1}{N}\sum_{i=1}^N
              \Expect\left[\left|f(X^i_t)-f(\bar{X}^i_t)\right|^2\mathds{1}_S\right]^{1/2}\\&\quad
                                                                                              \leq
                                                                                              \frac{\text{(const.)}}{N}\sum_{i=1}^N
	\Expect\left[\left\|X^i_t-\bar{X}^i_t\right\|_2^2\mathds{1}_S\right]^{1/2}
	%\\&
	\leq \frac{\text{(const.)}}{\sqrt{N}}
	\end{align*}
	where we used triangle inequality in the first step, the
	Lipschitz property of $f$ in the second step, and the
	estimate~\eqref{eq:estimate_POA} from Prop.~\ref{prop:prop-chaos} in the final step. 
	
\item (Bound~\eqref{eq:bound-prop-chaos-Sc}) The function $f$ is assumed bounded. So,
	\begin{align*}
	\Expect\left[\left|\frac{1}{N}\sum_{i=1}^N f(X^i_t) -\Expect[f(X_t) |\clZ_t]\right|^2\mathds{1}_{S^c}\right]^{1/2} \leq \text{(const)} \P(S^c)^\half
	\end{align*}
	The probability of the event $S^c$ is bounded as follows: 
	\begin{align*}
	\P(S^c) %&= \P(\SigN_0 \prec \frac{1}{2}\Sigma_0)\\
	& = \P(\Sigma_0 - \SigN_0 \succ \frac{1}{2}\Sigma_0)\\
	& \leq  \P( \| \SigN_0 - \Sigma_0\|_F^2 \geq \frac{1}{4} \|\Sigma_0\|_F^2)\\
	& \leq \frac{\Expect[ \| \SigN_0 - \Sigma_0\|_F^2]}{\frac{1}{4} \|\Sigma_0\|_F^2}\leq \frac{12\trace(\Sigma_0^2)}{N\|\Sigma_0\|_F^2} = \frac{12}{N}
	\end{align*}
%	concluding~\eqref{eq:bound-prop-chaos-Sc}.
	
\end{enumerate}	
\end{proof}

\section{Proof of the Prop.~\ref{prop:importance-sampling}}\label{apdx:importance-sampling}
\begin{proof}
	\noindent {\bf Part (i)}  Express the m.s.e as:
	\begin{equation*}
	\text{m.s.e}_{\overline{\text{PF}}}(f) = \Expect\left[\left|\frac{1}{N}\sum_{i=1}^N \bar{w}_i f(X^i_0) - \Expect[\bar{w}_i f( X^i_0)|\clZ_1]\right|^2 \right]
	\end{equation*}
	where we used $\pi(f) =
	\Expect[\bar{w}_if(X^i_0)|\clZ_1]$. The expectation simplifies
	to:
	\begin{equation*}
	\text{m.s.e}_{\overline{\text{PF}}} = \frac{1}{N}\left(\underbrace{\Expect[|\bar{w}_i f(X^i_0)|^2]}_{\text{1st term}} - \underbrace{\Expect[ |\Expect[\bar{w}_if(X^i_0)|\clZ_1]|^2]}_{\text{2nd term}}\right)
	\end{equation*}
	The two terms are simplified separately:
	\begin{enumerate}
		\item (2nd term) Note that
		$\Expect[\bar{w}_if(X^i_0)|\clZ_1]=\Expect[a^\top X_0 |\clZ_1] =
		\frac{\sigma_x^2}{\sigma_x^2+\sigma_w^2} a^\top Z_1 =
		\frac{1}{2}a^\top Z_1$ where we used
		$\sigma_x^2=\sigma_w^2=\sigma^2$ in the last step. Therefore, the
		(2nd term) is evaluated as
		\begin{equation*}
		\Expect [|\frac{1}{2}a^\top Z_1|^2] =  \frac{1}{4}a^\top \Expect[Z_1Z_1^\top]a = \frac{\sigma^2}{2}
		\end{equation*}
		where we used $\Expect[Z_1Z_1^\top] = \Expect[X_0X_0^\top]  + \sigma_w^2 I_d = 2\sigma^2 I_d$.
		\item (1st term) We have
		\begin{align*}
		\Expect[|\bar{w}_if(X^i_0)|^2] &= \Expect\left[\frac{|f(X^i_0)|^2e^{-\frac{\|Z_1-X^i_0\|_2^2}{\sigma_w^2}}}{\left|\Expect[e^{-\frac{\|Z_1-X^i_0\|_2^2}{2\sigma_w^2}}|\clZ_1]\right|^2}\right]
		%\\
		%&=\Expect[\frac{|X^i_0|^2e^{-\frac{|Y-X^i_0|^2}{\sigma^2}}}{\frac{1}{2^{d}} e^{-\frac{|Y|^2}{2\sigma^2}}}]
		\end{align*}
		The denominator
		\begin{align*}
		\left|\Expect[e^{-\frac{\|Z_1-X^i_0\|_2^2}{2\sigma_w^2}}|\clZ_1]\right|^2 &= \left|\frac{(2\pi\sigma_w^2)^{d/2}}{(2\pi(\sigma_w^2 +\sigma_0^2))^{d/2}}e^{-\frac{\|Z_1\|_2^2}{2(\sigma_w^2+\sigma_x^2)}}\right|^2\\&=\frac{1}{2^d}e^{-\frac{\|Z_1\|_2^2}{2\sigma^2}}
		\end{align*}
		The conditional expectation of the numerator
		\begin{align*}
		\Expect&[|f(X^i_0)|^2e^{-\frac{\|Z_1-X^i_0\|_2^2}{\sigma_w^2}}|\clZ_1] \\&= \frac{(\pi\sigma_w^2)^{d/2}}{(\pi(\sigma_w^2+2\sigma_0^2))^{d/2}}e^{-\frac{\|Z_1\|^2}{2\sigma_0^2+\sigma_w^2}}\Expect_{\zeta\sim \NN(\frac{2}{3}Z_1,\frac{\sigma^2}{3}I_d)}[|f(\zeta)|^2]\\&=\frac{1}{3^{d/2}}e^{-\frac{\|Z_1\|^2}{3\sigma^2}}(\frac{4}{9}|a^\top Z_1|^2 + \frac{\sigma^2}{3})
		\end{align*}
		Using the tower property of the conditional expectation
		\begin{align*}
		\Expect&[|\bar{w}_if(X^i_0)|^2=\frac{2^d}{3^{d/2}} \Expect[e^{\frac{\|Z_1\|^2}{6\sigma^2}}(\frac{4}{9}|a^\top Z_1|^2 + \frac{\sigma^2}{3})]\\
		&=\frac{2^d}{3^{d/2}} \frac{(12\pi\sigma^2)^{d/2}}{(4\pi\sigma^2)^{d/2}}\Expect_{\zeta\sim\NN(0,6\sigma^2I_d)}[\frac{4}{9}|a^\top \zeta|^2 + \frac{\sigma^2}{3}]\\
		&= 2^d (3\sigma^2)
		\end{align*}
	\end{enumerate}
	The two terms are combined to obtain the formula~\eqref{eq:PF-mse-bound}.  
	
	\medskip
	
	\noindent {\bf Part (ii)}  The FPF estimator is $\pi_\text{FPF}(f) = a^\top m^{(N)}_1 $ where
	\begin{equation*}
	\ud m^{(N)}_t = \KN_t(\ud Z_t -\mN_t \ud t)
	\end{equation*}
	where $\KN_t = \frac{1}{\sigma_w^2}\SigN_t$. 
	The exact mean evolves according to:
	\begin{equation*}
	\ud m_t = \K_t (\ud Z_t - m_t \ud t)
	\end{equation*}
	where $\K_t = \frac{1}{\sigma_w^2}\Sigma_t$. 
	Therefore, the difference $\mN_t - m_t$ solves the sde:
	\begin{align*}
	\ud \mN_t - \ud m_t = - \KN_t(\mN_t-m_t)\ud t + (\KN_t -\K_t) \ud I_t 
	\end{align*}
	where $\ud I_t = \ud Z_t -m_t\ud t$ is the increment of the
	innovation process. Let $\Phi_{t,s}$ be the state transition
	matrix for the linear system $\frac{\ud}{\ud t}x_t = - \KN_t
	x_t$.  In terms of this matrix
	\begin{align*}
	\mN_1 - m_1 =  \Phi_{1,0}(\mN_0-m_0) + \int_0^1 \Phi_{1,t}(\KN_t - \K_t)\ud I_t
	\end{align*}   
	Taking an inner product of both sides with $a$ yields
	\begin{align*}
	a^\top\mN_1 - a^\top m_1 =  &a^\top\Phi_{1,0}(\mN_0-m_0) \\&+ \int_0^1 a^\top\Phi_{1,t}(\KN_t - \K_t)\ud I_t
	\end{align*}   
	Therefore,
	\begin{align*}
	\Expect&[|a^\top\mN_1 - a^\top m_1|^2] \leq \\&2 \Expect[ |a^\top \Phi_{1,0}(\mN_0-m_0)|^2] +2 \Expect[(\int_0^1 a^\top \Phi_{1,t}(\KN_t - \K_t)\ud I_t)^2]
	\end{align*}   
	The formula~\eqref{eq:FPF-mse-bound} follows by showing the following bounds for the two terms:
	\begin{subequations}
		\begin{align}
		\Expect[ |a^\top \Phi_{1,0}(\mN_0-m_0)|^2]  &\leq  \frac{2d\sigma^2}{N} \label{eq:FPF-bound-1}\\
		\Expect[(\int_0^1 a^\top \Phi_{1,t}(\KN_t - \K_t)\ud I_t)^2] &\leq  \frac{3d^2\sigma^2}{N}  \label{eq:FPF-bound-2}%\\
		%	\Expect[2\|\KN\|_2^2 \|W^{(N)}\|_2^2] &\leq  \frac{2d\sigma^2}{N}  \label{eq:EnKF-bound-3}
		\end{align}
	\end{subequations} 
	\begin{enumerate}	
		\item (Bound~\eqref{eq:FPF-bound-1}) The spectral norm
		$\|\Phi_{t,s}\|_2\leq 1$ because $\KN_t = \frac{1}{\sigma2}\SigN_t
		\succeq0$.  Therefore, $|a^\top \Phi_{1,0}(\mN_0-m_0)| \leq \|a\|_2
		\|\Phi_{1,0}\|_2 \|(\mN_0-m_0)\|_2\leq \|(\mN_0-m_0)\|_2$ and
		\begin{align*}
		\Expect[ |a^\top \Phi_{1,0}(\mN_0-m_0)|^2] \leq 	\Expect[ \|\mN_0-m_0\|_2^2]= \frac{\sigma_0^2 d}{N}
		\end{align*} 
		\item (Bound~\eqref{eq:FPF-bound-2}) By the It\^o isometry, because the innovation process is a Brownian motion~\cite[Lemma 5.6]{xiong2008}, 
		\begin{align*}
		\Expect[&(\int_0^1 a^\top \Phi_{1,t}(\KN_t - \K_t)\ud I_t)^2]
		\\&= \sigma_w^2\Expect[\int_0^1 a^\top \Phi_{1,t}(\KN_t -
		\K_t)^2 \Phi_{1,t}^\top a  \ud t  ]\\
		&\leq \frac{1}{\sigma_w^2}\int_0^1 \Expect [\|\SigN_t - \Sigma_t)\|^2_2] \ud t  
		\end{align*}
		where we used $\|\Phi_{1,t}\|_2\leq 1$ and $\|a\|_2=1$ to derive the
		inequality.  
		
		Next, we bound the spectral norm $\|\SigN_t - \Sigma\|_2$.  We have
		\begin{align*}
		\frac{\ud}{\ud t}\SigN_t = -\frac{1}{\sigma_w^2}(\SigN_t)^2,\quad \frac{\ud}{\ud t}\Sigma_t = -\frac{1}{\sigma_w^2}\Sigma_t^2
		\end{align*}
		and thus
		\begin{align*}
		\frac{\ud}{\ud t}(\SigN_t-\Sigma_t) =& -\frac{1}{2\sigma_w^2}(\SigN_t+\Sigma_t)(\SigN_t-\Sigma_t)\\
		&-\frac{1}{2\sigma_w^2}(\SigN_t-\Sigma_t)(\SigN_t+\Sigma_t)\
		\end{align*}
		Its solution is obtained as
		\begin{align*}
		\SigN_t-\Sigma_t =  \Phi_t(\SigN_0-\Sigma_0)\Phi_t^\top
		\end{align*}
		where $\Phi_t$ solves $\frac{\ud}{\ud t}\Phi_t =
                -\frac{1}{2\sigma^2}(\Sigma_t+\SigN_t)\Phi_t$ with
                $\Phi_0=I$. Because $\Sigma_t+\SigN_t\succeq 0$, the
                spectral norm $\|\Phi_t\|_2 \leq 1$.  Therefore,
		\begin{align*}
		\|\SigN_t-\Sigma_t\|_2 \leq \|\SigN_0-\Sigma_0\|_2
		\end{align*}
and 
		\begin{align*}
		\frac{1}{\sigma_w^2}\int_0^1 \Expect\|\SigN_t - \Sigma_t)\|^2_2 \ud t   &\leq \frac{1}{\sigma_w^2} \|\SigN_0 - \Sigma_0\|_2 \\
		&\leq \frac{1}{\sigma_w^2} \|\SigN_0 - \Sigma_0\|_F\\
		&\leq \frac{1}{\sigma_w^2} \frac{1}{N}\Expect[\|X^i_0\|^4]\\
		&\leq \frac{1}{\sigma_w^2} \frac{3\sigma_0^4 d^2}{N}
		\end{align*}
	\end{enumerate}

\end{proof}

\bibliographystyle{IEEEtran}
\bibliography{TAC-OPT-FPF}

\begin{IEEEbiography}[{\includegraphics[width=1in,height=1.25in,clip,keepaspectratio]{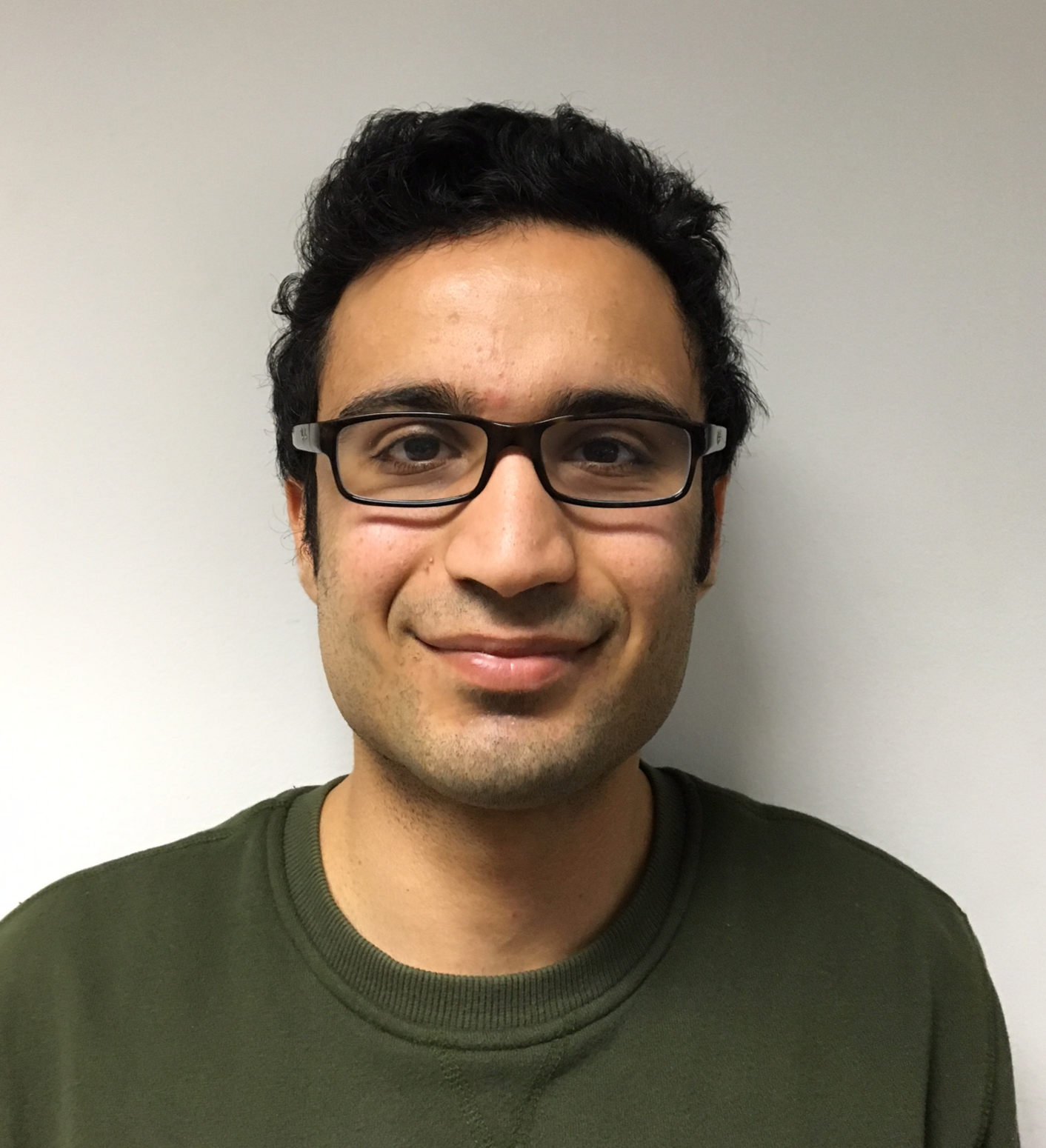}}]{Amirhossein Taghvaei}
	is a Postdoctoral Researcher in the Department of Mechanical and Aerospace Engineering at University of California Irvine. He obtained his PhD in Mechanical Science and Engineering and M.S in Mathematics from University of Illinois at Urbana-Champaign. He is currently working in the area of  Control theory and Machine learning.  
\end{IEEEbiography}

\vspace*{-66pt}

\begin{IEEEbiography}[{\includegraphics[width=1in,height=1.25in,clip,keepaspectratio]{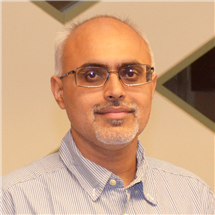}}]{Prashant G Mehta}
	(M'06) received the Ph.D. degree in Applied Mathematics from Cornell University, Ithaca, NY, in 2004. 
	He is an Associate Professor in the Department of Mechanical
	Science and Engineering, University of Illinois at
	Urbana-Champaign. Prior to joining Illinois, he was a Research
	Engineer at the United Technologies Research Center
	(UTRC).  His current research interests are in nonlinear filtering and
	mean-field games.  He received the Outstanding Achievement Award at
	UTRC for his contributions to the modeling and control of combustion instabilities in
	jet-engines. His students received the Best Student Paper Awards at
	the IEEE Conference on Decision and Control 2007 and 2009 and were
	finalists for these awards in 2010 and 2012.  He has served on the
	editorial boards of the {\em ASME Journal of Dynamic Systems, Measurement,
		and Control} and the {\em Systems and Control Letters}. 
\end{IEEEbiography}

\end{document}